\documentclass[11pt]{article}
\usepackage{setspace}
\usepackage[pdftex]{graphicx}
\usepackage{mathrsfs,amsmath,amsthm,amsfonts,color,ascmac,amssymb,bm}
\usepackage[dvipsnames]{xcolor}
\usepackage{natbib}
\usepackage{authblk}
\usepackage{booktabs}
\usepackage{enumitem}
\usepackage{type1cm} 
\usepackage{mathrsfs,comment}

\newtheorem{thm}{Theorem}
\newtheorem{eg}{Example}
\theoremstyle{definition}

\newcommand{\ep}{\varepsilon}

\newcommand{\dd}{\mathrm{d}}
\newcommand{\E}{\mathbb{E}}

\usepackage[top=20mm,bottom=27mm,left=30mm,right=30mm]{geometry}
\title{On default priors for robust Bayesian estimation with divergences}
\author[1]{Tomoyuki Nakagawa\footnote{The corresponding author. }}
\author[2]{Shintaro Hashimoto}

\affil[1]{Department of Information Sciences, Tokyo University of Science, Japan}
\affil[2]{Department of Mathematics, Hiroshima University, Japan}
\date{Last update: \today}

\allowdisplaybreaks
\begin{document}
\maketitle

\begin{abstract}

This paper presents objective priors for robust Bayesian estimation against outliers based on divergences. The minimum $\gamma$-divergence estimator is well-known to work well estimation against heavy contamination. The robust Bayesian methods by using quasi-posterior distributions based on divergences have been also proposed in recent years. In objective Bayesian framework, the selection of default prior distributions under such quasi-posterior distributions is an important problem. In this study, we provide some properties of reference and moment matching priors under the quasi-posterior distribution based on the $\gamma$-divergence. In particular, we show that the proposed priors are approximately robust under the condition on the contamination distribution without assuming any conditions on the contamination ratio. Some simulation studies are also presented. 

%

\end{abstract}


\medskip


\allowdisplaybreaks[1]

\section{Introduction}
\label{intro}

The problem of the robust parameter estimation against outliers has a long history. For example, \citet{huber2009robust} provide an excellent review of the classical robust estimation theory. It is well-known that the maximum likelihood estimator (MLE) is not robust against outliers because it is obtained by minimizing the Kullback-Leibler (KL) divergence between the true and empirical distributions. To overcome this problem, we may use other (robust) divergences instead of the KL divergence. The robust parameter estimation based on divergences  has been one of the central topics in modern robust statistics (e.g. \cite{Basu2011robust}). Such method was firstly proposed by \cite{basu1998robust}, who referred to it as the minimum density power divergence estimator. \cite{jones2001comparison} also proposed the ``type 0 divergence", which is a modified version of the density power divergence, and \cite{fujisawa2008robust} showed that it has good robustness properties. the type 0 divergence is also known as the $\gamma$-divergence, and statistical methods based on the $\gamma$-divergence have been presented by many authors (e.g. \cite{hirose2016gamma}, \cite{kawashima2017robust}, \cite{hirose-masuda}). 

In Bayesian statistics, the robustness against outliers is also an important issue, and divergence-based Bayesian methods have been proposed in recent years. Such methods are known as quasi-Bayes (or general Bayes) methods in some studies, and the corresponding posterior distributions are called quasi-posterior (or general posterior) distributions. To overcome the model misspecification problem (see \cite{ bissiri2016general}), the quasi-posterior distributions are based on a general loss function rather than the usual log-likelihood function. In general, such general loss functions may not depend on an assumed statistical model. However, in this study, we use loss functions that depend on the assumed model because we are interested in the robust estimation problem against outliers, that is, the model is not misspecified but data generating distribution is wrong. In other words, we use divergences or scoring rules as a loss function for the quasi-posterior distribution (see also  \cite{Hooker2014robust}, \cite{ghosh2016robust}, \cite{nakagawa2019robust}, \cite{jewson2018robust}, \cite{hashimoto2020gamma}).
For example, \cite{Hooker2014robust} used the Hellinger divergence, \cite{ghosh2016robust} used the density power divergence, and \cite{nakagawa2019robust} and \cite{hashimoto2020gamma} used the $\gamma$-divergence. In particular, the quasi-posterior distribution based on the $\gamma$-divergence is referred to as the $\gamma$-posterior in \cite{nakagawa2019robust}, and they showed that the $\gamma$-posterior has good robustness properties to overcome problems in \cite{ghosh2016robust}. 

 Although the selection of priors is an important issue in Bayesian statistics, we often have no prior information in some practical situations. In such cases, we may use priors called default or objective priors, and we should select an appropriate objective prior in a given context. In particular, we consider the reference and moment matching priors in this paper. The reference prior was firstly proposed by \cite{bernardo1979reference} and moment matching prior was proposed by \cite{ghosh2011moment}. However, such objective priors generally depend on an unknown data generating distribution when we cannot assume that the contamination ratio is approximately zero. For example, if we assume the $\varepsilon$-contamination model (see e.g. \cite{huber2009robust}) as a data generating distribution, many objective priors depend on unknown contamination ratio and unknown contamination distribution because these objective priors involve the expectations under the data generating distribution. Although \cite{giummole2019objective} derived some kinds of reference priors under the quasi-posterior distributions based on some kinds of scoring rules, they just discussed the robustness of such reference priors when the contamination ratio $\ep$ is approximately zero. Furthermore, their simulation studies largely depend on the assumption for the contamination ratio. In other words, they indirectly assume that the contamination ratio $\ep$ is approximately zero. The current study derives the moment matching priors under the quasi-posterior distribution in a similar way to \cite{ghosh2011moment}, and we show that the reference and moment matching priors based on the $\gamma$-divergence do not approximately depend on such unknown quantities under a certain assumption for the contamination distribution even if the contamination ratio is not small. 

The rest of this paper is organized as follows. In Section \ref{sec:robust}, we review robust Bayesian estimation based on divergences referring to some previous studies. We derive moment matching priors based on the quasi-posterior distribution using an asymptotic expansion of the quasi-posterior distribution given by \cite{giummole2019objective} in Section \ref{objective}. Furthermore, we show that reference and moment matching priors based on the $\gamma$-posterior do not depend on the contamination ratio and the contamination distribution.
In Section \ref{sec:sim}, we compare empirical bias and mean squared error of posterior means through some simulation studies. Some discussion about the selection of tuning parameters are also provided.

\section{Robust Bayesian estimation using divergences}
\label{sec:robust}
In this section, we review a framework of robust estimation in the seminal paper by Fujisawa and Eguchi \cite{fujisawa2008robust} and we introduce the robust Bayesian estimation using divergences. Let $X_1, \ldots, X_n$ be independent and identically distributed (iid) random variables according to a distribution $G$ with the probability density function $g$ on $\Omega$, and let $\bm{X}_n=(X_1,\dots, X_n)$. We assume the parametric model $f_{\bm{\theta}}=f(x,\bm{\theta})$ ($\bm{\theta} \in \Theta \subset \mathbb{R}^p$), and consider the estimation problem for $\bm{\theta}$. 

Then, the $\gamma$-divergence between two probability densities $g$ and $f$ is defined by 
\begin{align*}
D_{\gamma}(g,f_{\bm{\theta}})=& \frac{1}{\gamma(\gamma+1)} \log \int_{\Omega} g(x)^{1+\gamma} \dd x\\
& -\frac{1}{\gamma} \log \int_{\Omega} g(x) f_{\bm{\theta}}(x)^{\gamma} \dd x +\frac{1}{\gamma+1} \log \int_{\Omega} f_{\bm{\theta}}(x)^{1+\gamma} \dd x,  
\end{align*}
where $\gamma>0$ is a tuning parameter on robustness. We also define the $\gamma$-cross entropy as
\[d_{\gamma}(g,f_{\bm{\theta}})=-\frac{1}{\gamma} \log \int_{\Omega} g(x) f_{\bm{\theta}}(x)^{\gamma} \dd x +\frac{1}{\gamma+1} \log \int_{\Omega} f_{\bm{\theta}}(x)^{1+\gamma} \dd x\]
 (see \cite{jones2001comparison}, \cite{fujisawa2008robust}). 
 
\subsection{Framework of robustness}

Fujisawa and Eguchi \cite{fujisawa2008robust} introduced a new framework of robustness which is different from classical them. When some of the data values are regarded as outliers, we need a robust estimation procedure. Typically, an observation that takes a large value is regarded as an outlier. Under this convention, many robust parameter estimation procedures have been proposed to reduce the bias caused by an outlier. An influence function is one of methods to measure the sensitivity of models against outliers. It is known that the bias of an estimator is approximately proportional to the influence function when the contamination ratio $\ep$ is small. However, when $\ep$ is not small, the bias cannot be approximately proportional to the influence function. \cite{fujisawa2008robust} showed that the likelihood function based on the $\gamma$-divergence gives a sufficiently small bias under heavy contamination. Suppose that observations are generated from a mixture distribution $g(x)=(1-\ep) f(x) +\ep \delta(x)$, where $f(x)$ is the underlying density, $\delta(x)$ is another density function and $\ep$ is the contamination ratio. In Section \ref{objective}, we assume that the condition
\begin{align}
\nu_f = \left\{\int_{\Omega} \delta(x)f(x)^{\gamma_0} \dd x\right\}^{1/\gamma_0} \approx 0 \label{con_siki2}
\end{align}
holds for a constant $\gamma_0>0$ (see \cite{fujisawa2008robust}). When $x_0$ is generated from $\delta(x)$, we call $x_0$ outlier. We note that we do not assume that the contamination ratio $\ep$ is sufficiently small. This condition means that the contamination distribution $\delta(x)$ mostly lies on the tail of the underlying density $f(x)$. In other words, for an outlier $x_0$, it holds that $f(x_0) \approx 0$. We note that the condition \eqref{con_siki2} is also a basis to prove the robustness against outliers for the minimum $\gamma$-divergence estimator in \cite{fujisawa2008robust}. Furthermore, \cite{kanamori2014affine} provides some theoretical results of the $\gamma$-divergence, and related works in frequentist setting have been also developed (e.g. \cite{hirose2016gamma}, \cite{kawashima2017robust}, \cite{hirose-masuda}, and so on). 

The rest of this section, we make a brief review of the general Bayesian updating, and introduce some previous works which are closely related to this paper.

\subsection{General Bayesian updating}
We consider the same framework as \cite{ bissiri2016general} and \cite{jewson2018robust}. We are interested in $\bm{\theta}=\bm{\theta}(G)$ ($\bm{\theta} \in \Theta \subseteq \mathbb{R}^p$) and we define a loss function $\ell_{\bm{\theta}}(\bm{x}):=\ell(\bm{\theta}, \bm{x})$. Further, let $\bm{\theta}^*=\arg \min_{\bm{\theta} \in \Theta} \mathbb{E}_G\ell_{\bm{\theta}}(\bm{X})$ be the target parameter. We define the risk function by $\mathbb{E}_G\ell_{\bm{\theta}}(\bm{X})$, and its empirical risk is defined by $R_n(\bm{\theta})=(1/n)\sum_{i=1}^n \ell_{\bm{\theta}}(X_i)$. For the prior distribution $\pi(\theta)$, the quasi-posterior density is defined by
\begin{align*}
\pi_{n,\omega}(\bm{\theta})\propto  \exp \{-\omega n R_n(\bm{\theta})\} \pi(\bm{\theta}), 
\end{align*}
where $\omega>0$ is a tuning parameter called learning rate. We note that the quasi-posterior is also called general posterior or Gibbs posterior. In this paper, we fix $\omega=1$ for the same reason as \cite{jewson2018robust}. For example, if we set $\ell_{\mu}(x)=|x-\mu|$, we can estimate the median of distribution without assuming the statistical model. However, we consider the model-dependent loss function which is based on statistical divergence (or scoring rule) in this study (see also \cite{ghosh2016robust}, \cite{nakagawa2019robust}, \cite{jewson2018robust}, \cite{hashimoto2020gamma}). The unified framework of inference using the quasi-posterior distribution is discussed by \cite{ bissiri2016general}.


\subsection{ Assumptions and previous works}

Let $d(\cdot, \cdot)$ be a cross entropy induced by a divergence and let $\{f_{\bm{\theta}}: \bm{\theta} \in \Theta\}$ be a statistical model. In general, the quasi-posterior distribution based on a cross entropy is defined by
\begin{align}\label{quasi-post}
\pi^{(d)}(\bm{\theta}|\bm{X}_n)\propto \exp\left\{-n d(\bar{g}, f_{\bm{\theta}})\right\}\pi(\bm{\theta})=\exp\left\{ \sum_{i = 1}^n q^{(d)}(X_i; \bm{\theta})\right\}\pi(\bm{\theta}),
\end{align}
where $d(\bar{g}, f_{\bm{\theta}})$ is the empirically estimated cross entropy and $\bar{g}$ is the empirical density function. In robust statistics based on divergences, we may use the cross entropy induced by a robust divergence (e.g. \cite{basu1998robust}, \cite{jones2001comparison}, \cite{fujisawa2008robust}). In this paper, we mainly use the $\gamma$-cross entropy proposed by \cite{jones2001comparison} and \cite{fujisawa2008robust}. Recently, \cite{nakagawa2019robust} proposed the $\gamma$-posterior based on the monotone transformation of the $\gamma$-cross entropy
\begin{align*}
\tilde{d}_{\gamma}(g, f_{\bm{\theta}}) &= - \frac{1}{\gamma}\left\{\exp(-\gamma d_{\gamma}(g, f_{\bm{\theta}})) -1\right\} = -\frac{1}{\gamma}\frac{\int_{\Omega} g(x)f_{\bm{\theta}}(x)^{\gamma}\dd x}{\left(\int_{\Omega}f_{\bm{\theta}}(x)^{1+\gamma}\dd x\right)^{\gamma/(1+\gamma)}}+\frac{1}{\gamma} 
\end{align*}
for $\gamma>0$. The $\gamma$-posterior is defined by taking $d(\bar{g},f_{\bm{\theta}})=\tilde{d}_{\gamma}(\bar{g}, f_{\bm{\theta}})$ in \eqref{quasi-post}. 
On the other hand, \cite{ghosh2016robust} proposed the $R^{(\alpha)}$-posterior based on the density power cross entropy
\[
d_{\alpha} (g,f_{\bm{\theta}}) = -\frac{1}{\alpha}  \int g f_{\bm{\theta}}^{\alpha} \dd x +\frac{1}{1+\alpha} \log\int_{\Omega} f_{\bm{\theta}}^{1+\alpha} \dd x
\]
for $\alpha>0$. The $R^{(\alpha)}$-posterior is defined by taking $d(\bar{g},f_{\theta})=d_{\alpha}(\bar{g}, f_{\theta})$ in \eqref{quasi-post}. 
Note that cross entropies $d_{\alpha}(\cdot, \cdot)$ and $\tilde{d}_{\gamma}(\cdot, \cdot)$ converge to the negative log-likelihood function as $\alpha\to 0$ and $\gamma\to 0$, respectively. Hence, we can establish that they are some kinds of generalization of the negative log-likelihood function.  It is known that the posterior mean based on $R^{(\alpha)}$-posterior works well for the estimation of a location parameter in the presence of outliers. However, this is known to be unstable in the case of the estimation for a scale parameter (see \cite{nakagawa2019robust}). Nakagawa and Hashimoto \cite{nakagawa2019robust} showed that the posterior mean under the $\gamma$-posterior has small bias under heavy contamination for both location and scale parameters in some simulation studies. 

 Let $\bm{\theta}_g:=\arg \min_{\bm{\theta} \in \Theta} d(g, f_{\bm{\theta}})$ be the target parameter. We now assume the following regularity conditions on the density function $f_{\bm{\theta}}(x) = f(x; \bm{\theta}) \ (\bm{\theta} \in \Theta \subset \mathbb{R}^p)$. We use indices to denote derivatives of $\bar{D}(\bm{\theta}) = d(\bar{g}, f_{\bm{\theta}})$ with respect to the components of the parameter $\bm{\theta}$. For example, $\bar{D}_{ijk}(\bm{\theta}) = \partial_{i}\partial_{j}\partial_{k}\bar{D}(\bm{\theta})$ and $\bar{D}_{ijk\ell}(\bm{\theta}) = \partial_{i}\partial_{j}\partial_{k}\partial_{\ell}\bar{D}(\bm{\theta})$ for $i, j, k ,\ell = 1, \ldots, p$. 
\begin{enumerate}[label = (A\arabic*)]
\item The support of the density function does not depend on unknown parameter $\bm{\theta}$ and $f_{\bm{\theta}}$ is fifth-order differentiable with respect to $\bm{\theta}$ in neighbourhood $U$ of $\bm{\theta}_g$. \label{a1}
\item  Interchange of the order of integration with respect to $x$ and differentiation as $\bm{\theta}_g$ is justified. \label{a2} The expectations 
\[
\E_{g}[\partial_i\partial_j\partial_k q^{(d)}(X_1; {\bm{\theta}}_g)]\quad  \text{and} \quad  \E_{g}[\partial_i\partial_j\partial_k\partial_\ell q^{(d)}(X_1; {\bm{\theta}}_g)]
\] 
are all finite and $M_{ijk\ell s}(x)$ exists such that 
\begin{align*} 
\sup_{\bm{\theta} \in U} \left|\partial_i\partial_j\partial_k\partial_\ell\partial_s q^{(d)}(x; \bm{\theta})\right| \leq M_{ijk\ell s}(x)
\end{align*}
 and $\E_{g}\left[M_{ijk\ell s}(X_1)\right] < \infty$ for all $i, j, k, \ell, s = 1, \ldots, p$, where $\partial_i = \partial/\partial \theta_i$ and $\partial = \partial/\partial \bm{\theta}$, and $\E_g(\cdot)$ is expectation of $X$ with respect to a probability density function $g$.
\item For any $\delta > 0$, with probability one \label{a3}
\begin{align*}
\sup_{\| \bm{\theta} - \bm{\theta}_g \| > \delta } \left\{ d(\bar{g}, f_{\bm{\theta}_g}) - d(\bar{g}, f_{\bm{\theta}}) \right\} < -\ep 
\end{align*}
for some $\ep > 0$ and for all sufficiently large $n$. 
\end{enumerate}
The matrices $I^{(d)}(\bm{\theta})$ and $J^{(d)}(\bm{\theta})$ are defined by 
 \begin{align*}
&I^{(d)}(\bm{\theta}) = \E_{g}\left[\partial q^{(d)}( X_1; \bm{\theta}) \partial^{\top} q^{(d)}( X_1; \bm{\theta})\right], \\
& J^{(d)}(\bm{\theta}) = - \E_{g}\left[\partial \partial^{\top} q^{(d)}( X_1; \bm{\theta})\right],    
\end{align*}
respectively. We also assume that $I^{(d)}(\bm{\theta})$ and $J^{(d)}(\bm{\theta})$ are positive definite matrices. Under these conditions, \cite{ghosh2016robust} and \cite{nakagawa2019robust} discussed  several asymptotic properties of the quasi-posterior distributions and the corresponding posterior means.

In terms of the higher-order asymptotic theory, Giummol{\`e} et al.  \cite{giummole2019objective} derived the asymptotic expansion of such quasi-posterior distributions. We now introduce the notation that will be used in the rest of the paper. Then \cite{giummole2019objective} presented the following theorem.

\begin{thm}[Giummol{\`e} et al.  \cite{giummole2019objective}]\label{asymp3}
 Under the conditions {\rm \ref{a1}--\ref{a3}}, we assume that $\hat{\bm{\theta}}_n^{(d)}$ is a consistent solution of $\partial d(\bar{g}, f_{\bm{\theta}}) = \bm{0}$ and $\hat{\bm{\theta}}^{(d)}_n \xrightarrow{p} \bm{\theta}_g $ as $n \rightarrow \infty$. Then for any prior density function $\pi(\bm{\theta})$ that is third-order differentiable and positive at $\bm{\theta}_g$, it holds that
\begin{align}
\begin{split}\label{expansion1}
\pi^{*(d)}(\bm{t}_n|\bm{X}_n)=\phi\left(\bm{t}_n; \tilde{J}^{-1}\right)\left(1 + n^{-1/2}A_1(\bm{t}_n) + n^{-1}A_2(\bm{t}_n)\right) + O_p(n^{-3/2}) 
\end{split}
\end{align} 
where 
$\pi^{*(d)}(\bm{t}_n|\bm{X}_n)$ is the quasi-posterior density function of the normalized random variable $\bm{t}_n = (t_1, \ldots, t_p)^{\top} = \sqrt{n}(\bm{\theta}-\hat{\bm{\theta}}_n^{(d)})$ given $\bm{X}_n$, and $\phi(\cdot; A)$ is the density function of a p-variate normal distribution with zero mean vector and covariance matrix $A$, and $\tilde{J} = J^{(d)}(\hat{\bm{\theta}}^{(d)}_n)$, $\tilde{J}^{-1} = (\tilde{J}^{ij})$, and
\begin{align*}
A_1(\bm{t}_n) =&  \sum_{i= 1}^{p}\frac{\partial_i \pi(\hat{\bm{\theta}}_n^{(d)})}{\pi(\hat{\bm{\theta}}_n^{(d)})} t_i + \frac{1}{6}\sum_{i, j, k}\bar{D}_{ijk}(\hat{\bm{\theta}}_n^{(d)})t_it_jt_k, \\
A_2(\bm{t}_n) =& \sum_{i, j}\frac{1}{2}\frac{\partial_i\partial_j\pi(\hat{\bm{\theta}}_n^{(d)})}{\pi(\hat{\bm{\theta}}_n^{(d)})}(t_it_j - \tilde{J}^{ij}) -\sum_{i, j, k, \ell} \frac{1}{6}\frac{\partial_i \pi(\hat{\bm{\theta}}_n^{(d)})}{\pi(\hat{\bm{\theta}}_n^{(d)})} \bar{D}_{jk\ell}(\hat{\bm{\theta}}_n^{(d)})\left(t_it_jt_kt_\ell - 3\tilde{J}^{ij}\tilde{J}^{k\ell}\right)\\
& -  \sum_{i, j, k, \ell} \frac{1}{24} \bar{D}_{ijk\ell}(\hat{\bm{\theta}}_n^{(d)})\left(t_it_jt_kt_\ell - 3\tilde{J}^{ij}\tilde{J}^{k\ell}\right) \\
& + \sum_{i, j, k, h, g, f} \frac{1}{72} \bar{D}_{ijk}\bar{D}_{hgf}(2t_it_jt_kt_ht_gt_f - 15\tilde{J}^{ij}\tilde{J}^{kh}\tilde{J}^{gf}). 
\end{align*}
\end{thm}
\begin{proof}
The proof is given in the Appendix of \cite{giummole2019objective}. 
\end{proof}

As previously mentioned, quasi-posterior distributions depend on a cross entropy induced by a divergence and a prior distribution. If we have some information about unknown parameters $\bm{\theta}$, we can use a prior distribution that takes such prior information into account. However, in the absence of prior information, we often use prior distributions known as default or objective priors. \cite{giummole2019objective} proposed the reference prior for quasi-posterior distributions, which is a type of objective priors (see \cite{bernardo1979reference}). The reference prior $\pi_R$ is obtained by asymptotically maximizing the expected KL divergence between prior and posterior distributions. As a generalization of the reference prior, \cite{ghosh2011general} discussed such priors under a general divergence measure known as the $\alpha$-divergence (see also \cite{liu2014divergence}, \cite{hashimoto2021}). The reference prior under the $\alpha$-divergence is given by asymptotically maximizing the expected $\alpha$-divergence
\[
H(\pi) = \E[{\cal D}^{(\alpha)}(\pi^{(d)}(\bm{\theta}| \bm{X}_n), \pi(\bm{\theta}))],
\]
 where $\mathcal{D}^{(\alpha)}$ is the $\alpha$-divergence defined as  
\begin{align*}
 \mathcal{D}^{(\alpha)}(\pi^{(d)}(\bm{\theta}| \bm{X}_n), \pi(\bm{\theta}) ) = \frac{1}{\alpha(1-\alpha)} \int_{\Theta} \left\{1-\left(\frac{\pi(\bm{\theta})}{\pi^{(d)}(\bm{\theta}|\bm{X}_n)}\right)\right\}^{\alpha}\pi^{(d)}(\bm{\theta}|\bm{X}_n) \dd \bm{\theta}
\end{align*}
which corresponds to the KL divergence as $\alpha\to 0$, the Hellinger divergence for $\alpha = 1/2$, and the $\chi^2$-divergence for $\alpha = -1$. \cite{giummole2019objective} derived reference priors with the $\alpha$-divergence under quasi-posterior based on some kinds of proper scoring rules such as the Tsallis scoring rule and Hyv\"{a}rinen scoring rule. We note that the former rule is the same as the density power score of \cite{basu1998robust} with minor notational modifications. 

\begin{thm}[Giummol{\`e} et al.  \cite{giummole2019objective}]\label{reference}
When $|\alpha| < 1$, the reference prior that asymptotically maximizes the expected $\alpha$-divergence between quasi-posterior and prior distributions is given by 
\[
\pi_{R}(\bm{\theta}) \propto \det(J^{(d)}(\bm{\theta}))^{1/2}. 
\]
\end{thm}
The result of Theorem \ref{reference} is similar to that of \cite{ghosh2011general} and \cite{liu2014divergence}. Objective priors such as the above theorem are useful because they can be determined by the data generating model. However, such priors do not have a statistical guarantee when the model is misspecified such as Huber's $\varepsilon$-contamination model. In other words, the reference prior in Theorem \ref{reference} depends on data generating distribution $g$ because of $J^{(d)}(\bm{\theta}) = - \E_{g}\left[\partial \partial^{\top} q^{(d)}(X_1; \bm{\theta})\right]$, where $g(x) = (1- \varepsilon)f_{\bm{\theta}}(x) + \varepsilon \delta(x)$ when the contamination ratio $\ep$ is not small such as heavy contamination cases. We now consider some objective priors under the $\gamma$-posterior that is robust against such unknown quantities in the next section.

\section{Main results}
\label{objective}
In this section, we show our main results. Our contributions are as follows. We derive moment matching priors for quasi-posterior distributions (Theorem \ref{moment}). We prove that the proposed priors is robust under the condition on the tail of the contamination distribution (Theorem \ref{pri_rob}).

\subsection{Moment matching priors}
The moment matching priors proposed by \cite{ghosh2011moment} are priors that match the posterior mean and MLE up to the higher order (see also \cite{hashimoto2019moment}). In this section, we attempt to extend the results of \cite{ghosh2011moment} to the context of quasi-posterior distributions. Our goal is to identify a prior such that the difference between the quasi-posterior mean $\tilde{\bm{\theta}}^{(d)}_n$ and frequentist minimum divergence estimator $\hat{\bm{\theta}}^{(d)}_n$ converges to zero up to the order of $o(n^{-1})$. From Theorem \ref{asymp3}, we have the following theorem. 
\begin{thm} \label{moment}
Let $\tilde{\bm{\theta}}_n^{(d)} = (\tilde{\theta}_1, \ldots, \tilde{\theta}_p)$, $\hat{\bm{\theta}}_n^{(d)} = (\hat{\theta}_1, \ldots, \hat{\theta}_p)$ and $\bm{t}_n=(t_1,\dots, t_p)^{\top}=\sqrt{n}(\bm{\theta} -\hat{\bm{\theta}}_n^{(d)})$. Under the same assumptions as Theorem \ref{asymp3}, it holds that
\begin{align*}
\begin{split}
 n \left(\tilde{\theta}_{\ell}^{(d)} - \hat{\theta}_{\ell}^{(d)}\right)
\xrightarrow{p}   \sum_{i= 1}^{p}\frac{\partial_i \pi(\bm{\theta}_g)}{\pi(\bm{\theta}_g)}J^{i\ell} + \frac{1}{6}\sum_{i, j, k}g_{ijk}^{(d)}(\bm{\theta}_g) \left(J^{ij}J^{k\ell} + J^{ik}J^{j\ell} + J^{i\ell}J^{jk}\right)
\end{split}
\end{align*}
as $n\to \infty$, where $J = J^{(d)}(\bm{\theta}_g)$, $J^{-1} = (J^{ij})$ and $g_{ijk}^{(d)}(\bm{\theta}) = \E_g\left[\partial_i \partial_j\partial_k q^{(d)}(X_1; \bm{\theta})\right]$. Furthermore, if we set a prior which satisfies
\begin{align}\label{matching-eq}
 \frac{\partial_\ell \pi(\bm{\theta})}{\pi(\bm{\theta})} + \frac{1}{2}\sum_{i, j}g_{ij\ell}^{(d)}(\bm{\theta}) J^{ij}(\bm{\theta})  = 0
\end{align} 
for all $\ell = 1, \ldots, p$, then it holds that 
\[
n \left(\tilde{\theta}_{\ell}^{(d)} - \hat{\theta}_{\ell}^{(d)}\right) \xrightarrow{p} 0
\]
for $\ell = 1, \ldots, p$ as $n\to \infty$, where $\{J^{(d)}(\bm{\theta})\}^{-1} = (J^{ij}(\bm{\theta}))$. 
\end{thm}

Hereafter, the prior that satisfies the equation \eqref{matching-eq} up to the order of $o_p(n^{-1})$ for all $\ell = 1, \ldots, p$ is referred to as a moment matching prior and we denote it by $\pi_M$.

\begin{proof}
From the asymptotic expansion of the posterior density (\ref{expansion1}), we have the asymptotic expansion of the posterior mean for $\theta_{\ell}$ as
\begin{align}\label{thm3-1}
\begin{split}
\tilde{\theta}_\ell^{(d)} =& \int_{\Theta} \theta_{\ell} \pi^{(d)}(\bm{\theta}| \bm{X}_n) \dd \bm{\theta}\\
 =& \hat{\theta}_\ell^{(d)}+ \frac{1}{\sqrt{n}}\int_{ \mathbb{R}^p}  t_\ell\pi^{*(d)}(\bm{t}_n|\bm{X}_n) \dd \bm{t}_n\\
 =& \hat{\theta}_\ell^{(d)}  + \frac{1}{n}\int_{\mathbb{R}^p}t_\ell\phi\left(\bm{t}_n; \tilde{J}^{-1}\right)A_1(\bm{t}_n)\dd \bm{t}_n+ O_p(n^{-3/2})
 \end{split}
\end{align} 
for $\ell = 1, \ldots, p$. 
The integral in the above equation is calculated by
\begin{align}\label{thm3-2}
\begin{split}
 \int_{\mathbb{R}^p} t_\ell A_1(\bm{t}_n)\phi\left(\bm{t}_n; \tilde{J}^{-1}\right)\dd \bm{t}_n=& \sum_{i= 1}^{p}\frac{\partial_i \pi(\hat{\bm{\theta}}_n^{(d)})}{\pi(\hat{\bm{\theta}}_n^{(d)})}\int_{\mathbb{R}^p}  t_it_{\ell} \phi\left(\bm{t}_n; \tilde{J}^{-1}\right)\dd \bm{t}_n \\
 &+ \frac{1}{6}\sum_{i, j, k}\bar{D}_{ijk}(\hat{\bm{\theta}}_n^{(d)})\int_{\mathbb{R}^p} t_it_jt_kt_\ell\phi\left(\bm{t}_n; \tilde{J}^{-1}\right)\dd \bm{t}_n \\
=& \sum_{i= 1}^{p}\frac{\partial_i \pi(\hat{\bm{\theta}}_n^{(d)})}{\pi(\hat{\bm{\theta}}_n^{(d)})}\tilde{J}^{i\ell} \\
&+ \frac{1}{6}\sum_{i, j, k}\bar{D}_{ijk}(\hat{\bm{\theta}}_n^{(d)}) \left(\tilde{J}^{ij}{\tilde{J}}^{k\ell} + \tilde{J}^{ik}\tilde{J}^{j\ell} + \tilde{J}^{i\ell}\tilde{J}^{jk}\right) + o_p(1).
\end{split} 
\end{align}
From \eqref{thm3-1} and \eqref{thm3-2} we have 
\[\tilde{\theta}_\ell^{(d)}-\hat{\theta}_\ell^{(d)}= \sum_{i= 1}^{p}\frac{\partial_i \pi(\hat{\bm{\theta}}_n^{(d)})}{\pi(\hat{\bm{\theta}}_n^{(d)})}\tilde{J}^{i\ell}+ \frac{1}{6n}\sum_{i, j, k}\bar{D}_{ijk}(\hat{\bm{\theta}}_n^{(d)}) \left(\tilde{J}^{ij}\tilde{J}^{k\ell} + \tilde{J}^{ik}\tilde{J}^{j\ell} + \tilde{J}^{i\ell}\tilde{J}^{jk}\right)+ O_p(n^{-3/2})\]
for $\ell = 1, \ldots, p$.
By using the consistency of the estimator $\hat{\bm{\theta}}_n^{(d)}$, we then have the following asymptotic difference between $\hat{\theta}_\ell^{(d)}$ and $\hat{\theta}_\ell^{(d)}$: 
\begin{align*}
n \left(\tilde{\theta}_{\ell}^{(d)} - \hat{\theta}_{\ell}^{(d)}\right) \xrightarrow{p}   \sum_{i= 1}^{p}\frac{\partial_i \pi(\bm{\theta}_g)}{\pi(\bm{\theta}_g)}J^{i\ell} + \frac{1}{6}\sum_{i, j, k}g_{ijk}^{(d)}(\bm{\theta}_g) \left(J^{ij}J^{k\ell} + J^{ik}J^{j\ell} + J^{i\ell}J^{jk}\right)
\end{align*}
as $n\to \infty$ for $\ell = 1, \ldots, p$. 
\end{proof}
In general, it is not easy to obtain the moment matching priors explicitly. Two examples are given as follows. 
\begin{eg}
When $p = 1$, the moment matching prior is given by 
\[
\pi_M(\theta) = C\exp \left\{-\int^{\theta} \frac{g_3^{(d)}(t)}{2J^{(d)}(t)}  \dd t\right\}
\]
for a constant $C$, where $g_3$ is a third derivation of $g$. This prior is very similar to that of \cite{ghosh2011moment}, but the quantities $g_3^{(d)}(t)$ and $J^{(d)}(t)$ are different from it. 
\end{eg}
\begin{eg}
When $p = 2$, we put 
\[
u_{\ell}(\theta_1, \theta_2) = \sum_{i, j}g_{ij\ell}^{(d)}(\bm{\theta})J^{ij}(\bm{\theta}) \ \ (\ell = 1, 2), 
\]
where $\bm{\theta} = (\theta_1, \theta_2)^{\top}$. If $u_{\ell}(\theta_1, \theta_2)$ only depends on $\theta_{\ell}$ for all $\ell = 1, 2$, and does not depend on other parameters $\theta_k (k \neq \ell)$, we have
 \begin{align*}
 u_{1}(\theta_1, \theta_2) \equiv u_{1}(\theta_{1}), \ \  u_{2}(\theta_1, \theta_2) \equiv u_{2}(\theta_{2}). 
 \end{align*}
Then we can solve the differential equation give by \eqref{matching-eq}, and the moment matching prior is obtained by 
\[
\pi_M(\theta_1, \theta_2) \propto \exp \left\{-\frac{1}{2}\int^{\theta_1} u_{1}(t_{1})  \dd t_{1}\right\}\exp\left\{ -\frac{1}{2}\int^{\theta_2} u_{2}(t_{2})  \dd t_{2} \right\}. 
\]
\end{eg}


\subsection{Robustness of objective priors}
For data that may be heavily contaminated, we cannot assume that the contamination ratio $\ep$ is approximately zero. In general, reference and moment matching priors depend on the contamination ratio and distribution. Therefore, we cannot directly use such objective priors for the quasi-posterior distributions because the contamination ratio $\ep$ and the contamination distribution $\delta(x)$ are unknown. In this subsection, we prove that priors based on $\gamma$-divergence are robust against theses unknown quantities. In addition to \eqref{con_siki2}, we assume the following condition of the contamination distribution 
\begin{align}
\nu_{\bm{\theta}} = \left\{\int_{\Omega} \delta(x)f_{\bm{\theta}}(x)^{\gamma_0} \dd x\right\}^{1/\gamma_0} \approx 0
  \label{con_siki1}
\end{align}
for all $\bm{\theta} \in \Theta$ and an appropriately large constant $\gamma_0>0$ (see also \cite{fujisawa2008robust}). 
Note that the assumption \eqref{con_siki1} is also a basis to prove the robustness against outliers for the minimum $\gamma$-divergence estimator in \cite{fujisawa2008robust}. Then, we have the following theorem.  
\begin{thm} \label{pri_rob}
Assume the condition \eqref{con_siki1}. Let
\begin{align*}
q^{(\gamma)}(x; \bm{\theta}) := q^{(\tilde{d}_{\gamma})}(x; \bm{\theta}) = \frac{1}{\gamma}f_{\bm{\theta}}(x)^{\gamma}\left\{\int_{\Omega}f_{\bm{\theta}}(y)^{1+\gamma} \dd y\right\}^{-\gamma/(1+\gamma)}, 
\end{align*}
and let 
\begin{align*}
h_{ij}^{(\gamma)}(\bm{\theta}) &= -\E_{f_{\bm{\theta}}}\left[\partial_i\partial_j q^{(\gamma)}(X_1; \bm{\theta})\right], \\
\tilde{g}_{ijk}^{(\gamma)}(\bm{\theta}) & = \E_{f_{\bm{\theta}}}\left[\partial_i\partial_j\partial_k q^{(\gamma)}(X_1; \bm{\theta})\right].  
\end{align*}
Then, it holds that
\begin{align}
\begin{split}\label{result1}
J_{ij}^{(\gamma)}(\bm{\theta}) &= -\E_{g}\left[\partial_i\partial_j q^{(\gamma)}(X_1; \bm{\theta})\right] =(1-\ep)h_{ij}^{(\gamma)}(\bm{\theta}) + O(\ep\nu^{\gamma}) ,\\
g_{ijk}^{(\gamma)}(\bm{\theta}) &= \E_{g}\left[\partial_i\partial_j\partial_k q^{(\gamma)}(X_1; \bm{\theta})\right] = (1-\ep)\tilde{g}_{ijk}^{(\gamma)}(\bm{\theta})+ O(\ep\nu^{\gamma}),
\end{split}
\end{align}
for $\gamma + 1 \leq \gamma_0$, where $\nu := \max\{\nu_f, \sup_{\bm{\theta}\in \Theta} \nu_{\bm{\theta}} \}$. The notation $O(\ep\nu^{\gamma})$ is the same use of that of \cite{fujisawa2008robust}. Furthermore, from above results, the reference prior and the equation \eqref{matching-eq} are approximately given by
\begin{align}
\begin{split}\label{eq_pri4}
&\pi_R(\bm{\theta}) \propto \det\left(H^{(\gamma)}(\bm{\theta})\right)^{1/2}, \\
&\frac{\partial_\ell \pi(\bm{\theta})}{\pi(\bm{\theta})} + \frac{1}{2}\sum_{i, j}\tilde{g}_{ij\ell}^{(\gamma)}(\bm{\theta}) h^{ij}(\bm{\theta}) = 0,  
\end{split}
\end{align} 
where $H^{(\gamma)}(\bm{\theta}) = (h_{ij}^{(\gamma)}(\theta))$ and $\{H^{(\gamma)}(\bm{\theta})\}^{-1} = (h^{ij}(\bm{\theta}))$. 
\end{thm}

\begin{proof}
Put $\ell(x) = \log f_{\bm{\theta}}(x)$, $\ell_i(x) = \partial_i\log f_{\bm{\theta}}(x)$, $\ell_{ij}(x) = \partial_i\partial_j\log f_{\bm{\theta}}(x)$ and $\ell_{ijk}(x) = \partial_i\partial_j\partial_k\log f_{\bm{\theta}}(x)$. First, from H\"older's inequality and Lyapunonv's inequality, it holds that 
\begin{align}
\begin{split}\label{siki-apen}
&\int_{\Omega} \left|\delta(x)f_{\bm{\theta}}(x)^{\gamma}\ell_i(x)\right| \dd x \leq \nu^{\gamma}\left(\int_{\Omega} |\ell_i(x)|^{1+\gamma}\delta(x)\dd x\right)^{1/(1+\gamma)}, \\
&\int_{\Omega} \left|\delta(x)f_{\bm{\theta}}(x)^{\gamma}\ell_i(x)\ell_j(x)\right| \dd x \leq \nu^{\gamma}\left(\int_{\Omega} |\ell_i(x)\ell_j(x)|^{1+\gamma}\delta(x)\dd x\right)^{1/(1+\gamma)}, \\
&\int_{\Omega} \left|\delta(x)f_{\bm{\theta}}(x)^{\gamma}\ell_{ij}(x)\right|\dd x \leq \nu^{\gamma}\left(\int_{\Omega} |\ell_{ij}(x)|^{1+\gamma}\delta(x)\dd x\right)^{1/(1+\gamma)}, \\
&\int_{\Omega}\left|\delta(x)f_{\bm{\theta}}(x)^{\gamma}\ell_{ijk}(x)\right| \dd x\leq \nu^{\gamma}\left(\int_{\Omega} |\ell_{ijk}(x)|^{1+\gamma}\delta(x)\dd x\right)^{1/(1+\gamma)}, \\
& \int_{\Omega} \left|\delta(x)f_{\bm{\theta}}(x)^{\gamma}\ell_{ij}(x)\ell_k(x)\right| \dd x \leq \nu^{\gamma}\left(\int_{\Omega} |\ell_{ij}(x)\ell_k(x)|^{1+\gamma}\delta(x)\dd x\right)^{1/(1+\gamma)}, \\
&\int_{\Omega} \left|\delta(x)f_{\bm{\theta}}(x)^{\gamma}\ell_i(X)\ell_j(X_1)\ell_k(X_1)\right|\dd x \leq \nu^{\gamma}\left(\int_{\Omega} |\ell_i(x)\ell_j(x)\ell_k(x)|^{1+\gamma}\delta(x)\dd x\right)^{1/(1+\gamma)}
\end{split}
\end{align}
for $i, j, k = 1, \ldots, p$. Using (\ref{siki-apen}) and the results in Appendix, we have
\begin{align*}
\left| \int_{\Omega} \delta(x)\partial_i\partial_j q^{(\gamma)}(x; \bm{\theta})\dd x \right|  \leq &\|f_{\bm{\theta}}\|_{\gamma+1}^{-\gamma} \gamma\int_{\Omega} \left|\delta(x)f_{\bm{\theta}}(x)^{\gamma}\ell_i(x)\ell_j(x) \right| \dd x\\
&  + \|f_{\bm{\theta}}\|_{\gamma+1}^{-\gamma} \int_{\Omega} \left|\delta(x)f_{\bm{\theta}}(x)^{\gamma}\ell_{ij}(x) \right| \dd x \\
& + \gamma  S_i\|f_{\bm{\theta}}\|_{\gamma+1}^{-1-2\gamma} \int_{\Omega} \left|\delta(x)f_{\bm{\theta}}(x)^{\gamma}\ell_{j}(x) \right| \dd x\\
&  +  \gamma  S_j\|f_{\bm{\theta}}\|_{\gamma+1}^{-1-2\gamma}\int_{\Omega} \left|\delta(x)f_{\bm{\theta}}(x)^{\gamma}\ell_{i}(x) \right| \dd x  \\
& + \frac{(1+2\gamma)}{ \|f_{\bm{\theta}}\|_{\gamma+1}^{2+3\gamma}}S_iS_j\int_{\Omega} \left|\delta(x)f_{\bm{\theta}}(x)^{\gamma}\right| \dd x\\
& + \|f_{\bm{\theta}}\|_{\gamma+1}^{-1-2\gamma}\int_{\Omega} \left|\delta(x)f_{\bm{\theta}}(x)^{\gamma}\right| \dd x\int_{\Omega} f_{\bm{\theta}}(y)^{\gamma+1}s_{ij}(y)\dd y,\\
= & O(\nu^{\gamma}), \\
\end{align*}
where 
\begin{align*}
s_{ij}(y) = &(\gamma+1)\ell_{i}(y)\ell_{j}(y) + \ell_{ij}(y), \\ 
S_{i} = & \int_{\Omega}f_{\bm{\theta}}(y)^{\gamma+1}\ell_{i}(y) \dd y
\end{align*}
for $i, j = 1, \ldots, p$. Similarly, it also holds that
\begin{align*}
&\int \delta(x) \partial_i\partial_j\partial_k q^{(\gamma)}(x; \bm{\theta})\dd x= O(\nu^{\gamma})
\end{align*} 
for $i, j, k = 1, \ldots, p$. Since,  
\begin{align*}
J_{ij}^{(\gamma)}(\bm{\theta}) =& -\E_{g}\left[\partial_i\partial_j q^{(\gamma)}(X_1; \bm{\theta})\right]= -(1-\ep)h_{ij}^{(\gamma)}(\bm{\theta})  - \ep \int_{\Omega} \delta(x)\partial_i\partial_j q^{(\gamma)}(x; \bm{\theta})\dd x,\\
g_{ijk}^{(\gamma)}(\bm{\theta}) =& \E_{g}\left[\partial_i\partial_j\partial_k q^{(\gamma)}(X_1; \bm{\theta})\right]= (1-\ep)\tilde{g}_{ijk}^{(\gamma)}(\bm{\theta}) + \ep\int \delta(x) \partial_i\partial_j\partial_k q^{(\gamma)}(x; \bm{\theta})\dd x,
\end{align*}
the proof of \eqref{result1} is complete. It is also easily to see the result \eqref{eq_pri4} from \eqref{result1}.  
\end{proof}

 It should be noted that \eqref{result1} looks like results to Theorem 5.1 in \cite{fujisawa2008robust}. However, $q^{(\gamma)}(x; \bm{\theta})$ and its derivative functions are different formulae from that of \cite{fujisawa2008robust}, so that  the derivative functions and the proof of \eqref{result1} are given in Appendix. Theorem \ref{pri_rob} shows that expectations in the right-hand side of $J_{ij}^{(\gamma)}(\bm{\theta})$ and $g_{ijk}^{(\gamma)}(\bm{\theta})$ only depend on the underlying model $f_{\bm{\theta}}$, but do not depend on the contamination distribution. Furthermore, reference and moment matching priors for the $\gamma$-posterior are obtained by the parametric model $f_{\bm{\theta}}$, that is, these do not depend on the contamination ratio and the contamination distribution. For example, for a normal distribution $N(\mu, \sigma^2)$, reference and moment matching priors are given by 
\begin{align}
\begin{split}\label{gam_pri}
&\pi_{R}^{(\gamma)}(\mu, \sigma) = \sigma^{-3 + 1/(1+\gamma)} + O(\ep\nu^{\gamma}), \\
& \pi_{M}^{(\gamma)}(\mu, \sigma) = \sigma^{-(\gamma+7)/\{2(1+\gamma)\}}+ O(\ep\nu^{\gamma}). 
\end{split}
\end{align}
However, reference and moment matching priors under $R^{(\alpha)}$-posterior depend on unknown quantities in the data generating distribution unless $\varepsilon \approx0$, since $J_{ij}^{(\alpha)}(\bm{\theta})$ and $g_{ijk}^{(\alpha)}(\bm{\theta})$ have the following forms:
\begin{align*}
J_{ij}^{(\alpha)}(\bm{\theta}) =& -\E_{g}\left[\partial_i\partial_j q^{(\alpha)}(X_1; \bm{\theta})\right] \\
=& -(1-\ep)\E_{f_{\bm{\theta}}}\left[\partial_i\partial_j q^{(\alpha)}(X_1; \bm{\theta})\right] - \frac{\ep}{1+\alpha}\int_{\Omega}\partial_i\partial_jf_{\bm{\theta}}(x)^{1+\alpha} \dd x + O(\ep\nu^{\alpha}),\\
g_{ijk}^{(\alpha)}(\bm{\theta}) = &\E_{g}\left[\partial_i\partial_j\partial_k q^{(\alpha)}(X_1; \bm{\theta})\right]\\
 =& (1-\ep)\E_{f_{\bm{\theta}}}\left[\partial_i\partial_j\partial_k q^{(\alpha)}(X_1; \bm{\theta})\right]+  \frac{\ep}{1+\alpha}\int_{\Omega}\partial_i\partial_j\partial_k f_{\bm{\theta}}(x)^{1+\alpha} \dd x+ O(\ep\nu^{\alpha}), 
\end{align*}
where 
\begin{align*}
q^{(\alpha)}(x; \bm{\theta}) := q^{(d_{\alpha})}(x; \bm{\theta}) = \frac{1}{\alpha}f_{\bm{\theta}}(x)^{\alpha} - \frac{1}{1+\alpha}\int_{\Omega}f_{\bm{\theta}}(y)^{1+\alpha} \dd y. 
\end{align*}
The priors given by \eqref{gam_pri} can be practically used under the condition \eqref{con_siki1} even if the contamination ratio $\varepsilon$ is not small.

\section{Simulation studies}
\label{sec:sim}

\subsection{Setting and results}

We present performance of posterior means under reference and moment matching priors through some simulation studies. In this section, we assume that the parametric model is the normal distribution with mean $\mu$ and variance $\sigma^2$, and consider the joint estimation problem for $\mu$ and $\sigma^2$. We assume that the true values of $\mu$ and $\sigma^2$ are $0$ and $1$, respectively. We also assume that the contamination distribution is the normal distribution with mean $\nu$ and variance $1$. In other words, the data generating distribution is expressed by
\[g(x)= (1-\varepsilon) N(0,1) + \varepsilon N(\nu,1),\]
where $\varepsilon$ is the contamination ratio and $n$ is the sample size. 
 We compare the performances of estimators in terms of empirical bias and mean squared error (MSE) among three methods, which include the ordinary KL divergence-based posterior, $R^{(\alpha)}$-posterior and $\gamma$-posterior (our proposal). We also employ three prior distributions for $(\mu,\sigma)$, namely, (i) uniform prior, (ii) reference prior, and (iii) moment matching prior. 

Since exact calculations of posterior means are not easy, we use the importance sampling Monte Carlo algorithm using the proposal distributions $N(\bar{x}, s^2)$ for $\mu$ and ${\rm IG}(6, 5s)$ for $\sigma$ (the inverse gamma distribution with parameters $a$ and b  is denoted by $\mathrm{IG}(a,b)$), where $\bar{x} = n^{-1}\sum_{i = 1}^{n} x_i$ and $s^2 = (n-1)^{-1}\sum_{i = 1}^{n} (x_i -\bar{x})^2$ (for details of the importance sampling, see e.g. \cite{robert2004monte}). We carry out the importance sampling with 10,000 steps and we compute empirical bias and MSE for posterior means $(\hat{\mu},\hat{\sigma})$ of $(\mu,\sigma)$ by 10,000 iterations. The simulation results are reported in Tables  \ref{sim:bias_mu} to \ref{sim:mse_var}. Reference and the moment matching priors for the $\gamma$-posterior are given by (\ref{gam_pri}), and those of $R^{(\alpha)}$-posterior are ``formally" given as follows: 
\begin{align}\label{power_pri}
\pi_{M}^{(\alpha)}(\mu, \sigma) \propto \sigma^{-2-\alpha}, \ \ \pi_{M}^{(\alpha)}(\mu, \sigma) \propto \sigma^{C_M/2}, 
\end{align}
where $C_M$ is a constant given by
\begin{align*}
C_M = -\frac{2+\alpha^2}{(1+\alpha)} + \frac{\alpha(1+\alpha)^3(2+ \alpha) + (10 -\alpha^2(-2+ \alpha(5 + \alpha(3 + \alpha))))\pi^{\alpha/2}}{(1+\alpha)(-\alpha(1+\alpha)^2 + (-2 + \alpha + \alpha^2 + \alpha^3)\pi^{\alpha/2})}. 
\end{align*}
The term ``formally" means that since reference and the moment matching priors for $R^{(\alpha)}$-posterior strictly depend on an unknown contamination ratio and contamination distribution, we set $\varepsilon= 0$ in these priors. On the other hand, our proposed objective priors do not need such an assumption, but we assume only the condition \eqref{con_siki1}. We note that \cite{giummole2019objective} also use the same formal reference prior in their simulation studies.


The simulation results of empirical bias and MSE of posterior means of $\mu$ and $\sigma$ are provided by Tables \ref{sim:bias_mu} to \ref{sim:mse_var}. We consider three prior distributions for $(\mu,\sigma)$, namely, uniform, reference and moment matching priors. In these tables, we set $\nu = 6$, $\ep = 0.00, 0.05, 0.20$ and $n=20, 50, 100$. We also set the tuning parameters for $R^{(\alpha)}$- and $\gamma$-posteriors as $0.2, 0.3, 0.5, 0.7$. 

Tables \ref{sim:bias_mu} and \ref{sim:mse_mu} show empirical bias and MSE of posterior means of mean parameter $\mu$ based on standard posterior, $R^{(\alpha)}$- and $\gamma$-posteriors. Empirical bias and MSE for two robust methods are smaller than those of standard  posterior mean (denoted by ``Bayes" in Tables \ref{sim:bias_mu} to \ref{sim:mse_var}) in the presence of outliers for a large sample size. When there is no outliers ($\ep=0$), it seems that three methods are comparable. On the other hand, when $\ep=0.05$ and $\ep=0.20$, the standard posterior mean get worse, while the performances of posterior means based on $R^{(\alpha)}$-posterior and $\gamma$-posterior are comparable for both empirical bias and MSE. 



We also presented the results of the estimation for variance parameter $\sigma$ in Tables \ref{sim:bias_var} and \ref{sim:mse_var}. When there is no outliers, the performances of robust Bayes estimators under uniform prior are slightly worse. On the other hand, reference and moment matching priors provide relatively reasonable results even if the sample size is small and $\ep=0$. Empirical bias and MSE of $R^{(\alpha)}$-posterior and $\gamma$-posterior means for $\alpha, \gamma = 0.5, 0.7$ remain small even if the contamination ratio $\ep$ is not small. In particular, empirical bias and MSE of $\gamma$-posterior means for $\sigma$ are shown to be drastically smaller than those of $R^{(\alpha)}$-posterior.






\begin{table}[!tbp]
\caption{Empirical biases of posterior means for $\mu$\label{sim:bias_mu}} 
\begin{center}
\resizebox{1\textwidth}{!}{
\begin{tabular}{rrcrcrrrrcrrrr}
\toprule
\multicolumn{2}{c}{\bfseries }&\multicolumn{1}{c}{\bfseries }&\multicolumn{1}{c}{\bfseries Bayes}&\multicolumn{1}{c}{\bfseries }&\multicolumn{4}{c}{\bfseries $R^{(\alpha)}$-posterior}&\multicolumn{1}{c}{\bfseries }&\multicolumn{4}{c}{\bfseries $\gamma$-posterior}\tabularnewline
\cline{4-4} \cline{6-9} \cline{11-14}
\multicolumn{1}{c}{$\varepsilon$}&\multicolumn{1}{c}{$n$}&\multicolumn{1}{c}{}&\multicolumn{1}{c}{$\alpha,\gamma \to 0.0$}&\multicolumn{1}{c}{}&\multicolumn{1}{c}{$\alpha=0.2$}&\multicolumn{1}{c}{$\alpha=0.3$}&\multicolumn{1}{c}{$\alpha=0.5$}&\multicolumn{1}{c}{$\alpha=0.7$}&\multicolumn{1}{c}{}&\multicolumn{1}{c}{$\gamma=0.2$}&\multicolumn{1}{c}{$\gamma=0.3$}&\multicolumn{1}{c}{$\gamma=0.5$}&\multicolumn{1}{c}{$\gamma=0.7$}\tabularnewline
\midrule
\multicolumn{14}{l}{\bfseries Uniform prior}\tabularnewline
$0.00$&$ 20$&&$-0.002$&&$-0.003$&$-0.003$&$-0.002$&$ 0.001$&&$-0.003$&$-0.003$&$-0.003$&$-0.002$\tabularnewline
$0.00$&$ 50$&&$-0.002$&&$-0.001$&$-0.001$&$-0.001$&$ 0.000$&&$-0.001$&$-0.001$&$-0.001$&$ 0.000$\tabularnewline
$0.00$&$100$&&$ 0.000$&&$ 0.000$&$ 0.000$&$ 0.000$&$ 0.001$&&$ 0.000$&$ 0.000$&$ 0.000$&$ 0.001$\tabularnewline
$0.05$&$ 20$&&$ 0.298$&&$ 0.109$&$ 0.075$&$ 0.098$&$ 0.172$&&$ 0.104$&$ 0.064$&$ 0.046$&$ 0.060$\tabularnewline
$0.05$&$ 50$&&$ 0.301$&&$ 0.053$&$ 0.020$&$ 0.009$&$ 0.016$&&$ 0.051$&$ 0.017$&$ 0.004$&$ 0.002$\tabularnewline
$0.05$&$100$&&$ 0.301$&&$ 0.038$&$ 0.012$&$ 0.004$&$ 0.002$&&$ 0.036$&$ 0.011$&$ 0.003$&$ 0.001$\tabularnewline
$0.20$&$ 20$&&$ 1.192$&&$ 0.917$&$ 0.800$&$ 0.815$&$ 0.973$&&$ 0.908$&$ 0.755$&$ 0.596$&$ 0.615$\tabularnewline
$0.20$&$ 50$&&$ 1.198$&&$ 0.869$&$ 0.638$&$ 0.362$&$ 0.478$&&$ 0.864$&$ 0.600$&$ 0.215$&$ 0.112$\tabularnewline
$0.20$&$100$&&$ 1.201$&&$ 0.862$&$ 0.578$&$ 0.158$&$ 0.108$&&$ 0.859$&$ 0.537$&$ 0.065$&$ 0.015$\tabularnewline[10pt]
\multicolumn{14}{l}{\bfseries Reference prior}\tabularnewline
$0.00$&$ 20$&&$-0.002$&&$-0.003$&$-0.004$&$-0.004$&$-0.003$&&$-0.003$&$-0.004$&$-0.004$&$-0.004$\tabularnewline
$0.00$&$ 50$&&$-0.002$&&$-0.001$&$-0.001$&$-0.001$&$ 0.000$&&$-0.001$&$-0.001$&$-0.001$&$ 0.000$\tabularnewline
$0.00$&$100$&&$ 0.000$&&$ 0.000$&$ 0.000$&$ 0.000$&$ 0.001$&&$ 0.000$&$ 0.000$&$ 0.000$&$ 0.001$\tabularnewline
$0.05$&$ 20$&&$ 0.298$&&$ 0.072$&$ 0.033$&$ 0.016$&$ 0.018$&&$ 0.070$&$ 0.030$&$ 0.010$&$ 0.006$\tabularnewline
$0.05$&$ 50$&&$ 0.301$&&$ 0.041$&$ 0.013$&$ 0.002$&$ 0.001$&&$ 0.040$&$ 0.011$&$ 0.001$&$-0.001$\tabularnewline
$0.05$&$100$&&$ 0.301$&&$ 0.033$&$ 0.010$&$ 0.003$&$ 0.001$&&$ 0.032$&$ 0.009$&$ 0.002$&$ 0.001$\tabularnewline
$0.20$&$ 20$&&$ 1.192$&&$ 0.808$&$ 0.558$&$ 0.295$&$ 0.293$&&$ 0.803$&$ 0.537$&$ 0.227$&$ 0.152$\tabularnewline
$0.20$&$ 50$&&$ 1.198$&&$ 0.820$&$ 0.504$&$ 0.143$&$ 0.079$&&$ 0.817$&$ 0.473$&$ 0.085$&$ 0.023$\tabularnewline
$0.20$&$100$&&$ 1.201$&&$ 0.838$&$ 0.495$&$ 0.071$&$ 0.027$&&$ 0.836$&$ 0.457$&$ 0.029$&$ 0.006$\tabularnewline[10pt]
\multicolumn{14}{l}{\bfseries Moment Matching prior}\tabularnewline
$0.00$&$ 20$&&$-0.002$&&$-0.003$&$-0.004$&$-0.004$&$-0.003$&&$-0.003$&$-0.004$&$-0.004$&$-0.004$\tabularnewline
$0.00$&$ 50$&&$-0.002$&&$-0.001$&$-0.001$&$-0.001$&$ 0.000$&&$-0.001$&$-0.001$&$-0.001$&$ 0.000$\tabularnewline
$0.00$&$100$&&$ 0.000$&&$ 0.000$&$ 0.000$&$ 0.001$&$ 0.001$&&$ 0.000$&$ 0.000$&$ 0.000$&$ 0.001$\tabularnewline
$0.05$&$ 20$&&$ 0.298$&&$ 0.059$&$ 0.025$&$ 0.010$&$ 0.008$&&$ 0.059$&$ 0.024$&$ 0.009$&$ 0.007$\tabularnewline
$0.05$&$ 50$&&$ 0.301$&&$ 0.037$&$ 0.011$&$ 0.002$&$-0.001$&&$ 0.036$&$ 0.010$&$ 0.001$&$-0.001$\tabularnewline
$0.05$&$100$&&$ 0.301$&&$ 0.031$&$ 0.009$&$ 0.002$&$ 0.001$&&$ 0.030$&$ 0.009$&$ 0.002$&$ 0.001$\tabularnewline
$0.20$&$ 20$&&$ 1.192$&&$ 0.759$&$ 0.486$&$ 0.220$&$ 0.196$&&$ 0.759$&$ 0.481$&$ 0.210$&$ 0.165$\tabularnewline
$0.20$&$ 50$&&$ 1.198$&&$ 0.799$&$ 0.462$&$ 0.111$&$ 0.043$&&$ 0.797$&$ 0.441$&$ 0.079$&$ 0.025$\tabularnewline
$0.20$&$100$&&$ 1.201$&&$ 0.828$&$ 0.468$&$ 0.058$&$ 0.018$&&$ 0.827$&$ 0.435$&$ 0.028$&$ 0.006$\tabularnewline
\bottomrule
\end{tabular}
}
\end{center}
\end{table}

\begin{table}[!tbp]
\caption{Empirical biases of posterior means for $\sigma$\label{sim:bias_var}} 
\begin{center}
\resizebox{1\textwidth}{!}{
\begin{tabular}{rrcrcrrrrcrrrr}
\toprule
\multicolumn{2}{c}{\bfseries }&\multicolumn{1}{c}{\bfseries }&\multicolumn{1}{c}{\bfseries Bayes}&\multicolumn{1}{c}{\bfseries }&\multicolumn{4}{c}{\bfseries $R^{(\alpha)}$-posterior}&\multicolumn{1}{c}{\bfseries }&\multicolumn{4}{c}{\bfseries $\gamma$-posterior}\tabularnewline
\cline{4-4} \cline{6-9} \cline{11-14}
\multicolumn{1}{c}{$\varepsilon$}&\multicolumn{1}{c}{$n$}&\multicolumn{1}{c}{}&\multicolumn{1}{c}{$\alpha,\gamma \to 0.0$}&\multicolumn{1}{c}{}&\multicolumn{1}{c}{$\alpha=0.2$}&\multicolumn{1}{c}{$\alpha=0.3$}&\multicolumn{1}{c}{$\alpha=0.5$}&\multicolumn{1}{c}{$\alpha=0.7$}&\multicolumn{1}{c}{}&\multicolumn{1}{c}{$\gamma=0.2$}&\multicolumn{1}{c}{$\gamma=0.3$}&\multicolumn{1}{c}{$\gamma=0.5$}&\multicolumn{1}{c}{$\gamma=0.7$}\tabularnewline
\midrule
\multicolumn{14}{l}{\bfseries Uniform prior}\tabularnewline
$0.00$&$ 20$&&$ 0.058$&&$ 0.148$&$ 0.225$&$ 0.733$&$ 2.089$&&$ 0.136$&$ 0.184$&$ 0.330$&$ 0.620$\tabularnewline
$0.00$&$ 50$&&$ 0.022$&&$ 0.049$&$ 0.067$&$ 0.122$&$ 0.263$&&$ 0.046$&$ 0.058$&$ 0.085$&$ 0.116$\tabularnewline
$0.00$&$100$&&$ 0.011$&&$ 0.024$&$ 0.031$&$ 0.053$&$ 0.088$&&$ 0.022$&$ 0.028$&$ 0.039$&$ 0.051$\tabularnewline
$0.05$&$ 20$&&$ 0.669$&&$ 0.438$&$ 0.476$&$ 1.620$&$ 4.335$&&$ 0.404$&$ 0.370$&$ 0.540$&$ 1.109$\tabularnewline
$0.05$&$ 50$&&$ 0.660$&&$ 0.203$&$ 0.144$&$ 0.188$&$ 0.475$&&$ 0.189$&$ 0.116$&$ 0.110$&$ 0.139$\tabularnewline
$0.05$&$100$&&$ 0.652$&&$ 0.134$&$ 0.078$&$ 0.087$&$ 0.135$&&$ 0.123$&$ 0.061$&$ 0.049$&$ 0.058$\tabularnewline
$0.20$&$ 20$&&$ 1.732$&&$ 1.848$&$ 2.086$&$ 5.500$&$ 9.627$&&$ 1.769$&$ 1.727$&$ 2.207$&$ 3.833$\tabularnewline
$0.20$&$ 50$&&$ 1.653$&&$ 1.558$&$ 1.304$&$ 1.098$&$ 3.158$&&$ 1.533$&$ 1.182$&$ 0.573$&$ 0.454$\tabularnewline
$0.20$&$100$&&$ 1.626$&&$ 1.508$&$ 1.151$&$ 0.506$&$ 0.563$&&$ 1.495$&$ 1.042$&$ 0.198$&$ 0.113$\tabularnewline[10pt]
\multicolumn{14}{l}{\bfseries Reference prior}\tabularnewline
$0.00$&$ 20$&&$-0.001$&&$ 0.009$&$ 0.006$&$-0.007$&$-0.013$&&$ 0.007$&$-0.001$&$-0.041$&$-0.117$\tabularnewline
$0.00$&$ 50$&&$ 0.000$&&$ 0.003$&$ 0.002$&$-0.004$&$-0.010$&&$ 0.003$&$ 0.000$&$-0.012$&$-0.036$\tabularnewline
$0.00$&$100$&&$ 0.000$&&$ 0.002$&$ 0.001$&$-0.002$&$-0.006$&&$ 0.002$&$ 0.000$&$-0.005$&$-0.016$\tabularnewline
$0.05$&$ 20$&&$ 0.576$&&$ 0.173$&$ 0.093$&$ 0.066$&$ 0.097$&&$ 0.161$&$ 0.069$&$ 0.000$&$-0.051$\tabularnewline
$0.05$&$ 50$&&$ 0.625$&&$ 0.119$&$ 0.050$&$ 0.028$&$ 0.029$&&$ 0.110$&$ 0.035$&$-0.003$&$-0.030$\tabularnewline
$0.05$&$100$&&$ 0.635$&&$ 0.096$&$ 0.039$&$ 0.024$&$ 0.026$&&$ 0.088$&$ 0.026$&$ 0.000$&$-0.014$\tabularnewline
$0.20$&$ 20$&&$ 1.580$&&$ 1.281$&$ 0.954$&$ 0.659$&$ 0.697$&&$ 1.258$&$ 0.877$&$ 0.427$&$ 0.303$\tabularnewline
$0.20$&$ 50$&&$ 1.598$&&$ 1.367$&$ 0.917$&$ 0.375$&$ 0.324$&&$ 1.354$&$ 0.832$&$ 0.181$&$ 0.071$\tabularnewline
$0.20$&$100$&&$ 1.599$&&$ 1.421$&$ 0.937$&$ 0.241$&$ 0.196$&&$ 1.413$&$ 0.839$&$ 0.068$&$ 0.014$\tabularnewline[10pt]
\multicolumn{14}{l}{\bfseries Moment Matching prior}\tabularnewline
$0.00$&$ 20$&&$-0.039$&&$-0.036$&$-0.044$&$-0.083$&$-0.186$&&$-0.034$&$-0.039$&$-0.061$&$-0.090$\tabularnewline
$0.00$&$ 50$&&$-0.015$&&$-0.014$&$-0.016$&$-0.029$&$-0.067$&&$-0.013$&$-0.014$&$-0.019$&$-0.027$\tabularnewline
$0.00$&$100$&&$-0.007$&&$-0.006$&$-0.007$&$-0.014$&$-0.032$&&$-0.006$&$-0.006$&$-0.008$&$-0.012$\tabularnewline
$0.05$&$ 20$&&$ 0.516$&&$ 0.093$&$ 0.021$&$-0.029$&$-0.113$&&$ 0.089$&$ 0.016$&$-0.021$&$-0.023$\tabularnewline
$0.05$&$ 50$&&$ 0.601$&&$ 0.089$&$ 0.026$&$-0.002$&$-0.037$&&$ 0.083$&$ 0.017$&$-0.011$&$-0.021$\tabularnewline
$0.05$&$100$&&$ 0.623$&&$ 0.082$&$ 0.027$&$ 0.010$&$-0.005$&&$ 0.075$&$ 0.017$&$-0.003$&$-0.010$\tabularnewline
$0.20$&$ 20$&&$ 1.481$&&$ 1.097$&$ 0.736$&$ 0.395$&$ 0.225$&&$ 1.094$&$ 0.717$&$ 0.373$&$ 0.361$\tabularnewline
$0.20$&$ 50$&&$ 1.559$&&$ 1.293$&$ 0.808$&$ 0.276$&$ 0.165$&&$ 1.287$&$ 0.748$&$ 0.162$&$ 0.084$\tabularnewline
$0.20$&$100$&&$ 1.579$&&$ 1.386$&$ 0.872$&$ 0.197$&$ 0.135$&&$ 1.381$&$ 0.787$&$ 0.061$&$ 0.019$\tabularnewline
\bottomrule
\end{tabular}
}
\end{center}
\end{table}

\begin{table}[!tbp]
\caption{Empirical MSEs of posterior means for $\mu$\label{sim:mse_mu}} 
\begin{center}
\resizebox{1\textwidth}{!}{
\begin{tabular}{rrcrcrrrrcrrrr}
\toprule
\multicolumn{2}{c}{\bfseries }&\multicolumn{1}{c}{\bfseries }&\multicolumn{1}{c}{\bfseries Bayes}&\multicolumn{1}{c}{\bfseries }&\multicolumn{4}{c}{\bfseries $R^{(\alpha)}$-posterior}&\multicolumn{1}{c}{\bfseries }&\multicolumn{4}{c}{\bfseries $\gamma$-posterior}\tabularnewline
\cline{4-4} \cline{6-9} \cline{11-14}
\multicolumn{1}{c}{$\varepsilon$}&\multicolumn{1}{c}{$n$}&\multicolumn{1}{c}{}&\multicolumn{1}{c}{$\alpha,\gamma \to 0.0$}&\multicolumn{1}{c}{}&\multicolumn{1}{c}{$\alpha=0.2$}&\multicolumn{1}{c}{$\alpha=0.3$}&\multicolumn{1}{c}{$\alpha=0.5$}&\multicolumn{1}{c}{$\alpha=0.7$}&\multicolumn{1}{c}{}&\multicolumn{1}{c}{$\gamma=0.2$}&\multicolumn{1}{c}{$\gamma=0.3$}&\multicolumn{1}{c}{$\gamma=0.5$}&\multicolumn{1}{c}{$\gamma=0.7$}\tabularnewline
\midrule
\multicolumn{14}{l}{\bfseries Uniform prior}\tabularnewline
$0.00$&$ 20$&&$0.050$&&$0.051$&$0.053$&$0.090$&$0.282$&&$0.051$&$0.053$&$0.057$&$0.078$\tabularnewline
$0.00$&$ 50$&&$0.020$&&$0.021$&$0.022$&$0.023$&$0.027$&&$0.021$&$0.022$&$0.023$&$0.025$\tabularnewline
$0.00$&$100$&&$0.010$&&$0.010$&$0.011$&$0.012$&$0.013$&&$0.010$&$0.011$&$0.012$&$0.013$\tabularnewline
$0.05$&$ 20$&&$0.223$&&$0.098$&$0.081$&$0.280$&$1.081$&&$0.096$&$0.075$&$0.076$&$0.159$\tabularnewline
$0.05$&$ 50$&&$0.144$&&$0.031$&$0.025$&$0.025$&$0.039$&&$0.031$&$0.025$&$0.025$&$0.027$\tabularnewline
$0.05$&$100$&&$0.118$&&$0.015$&$0.012$&$0.013$&$0.013$&&$0.014$&$0.012$&$0.013$&$0.014$\tabularnewline
$0.20$&$ 20$&&$1.761$&&$1.267$&$1.127$&$2.296$&$4.781$&&$1.254$&$1.031$&$0.906$&$1.402$\tabularnewline
$0.20$&$ 50$&&$1.571$&&$0.950$&$0.647$&$0.311$&$0.879$&&$0.944$&$0.613$&$0.188$&$0.088$\tabularnewline
$0.20$&$100$&&$1.509$&&$0.844$&$0.494$&$0.095$&$0.052$&&$0.840$&$0.463$&$0.046$&$0.019$\tabularnewline[10pt]
\multicolumn{14}{l}{\bfseries Reference prior}\tabularnewline
$0.00$&$ 20$&&$0.050$&&$0.052$&$0.054$&$0.062$&$0.077$&&$0.052$&$0.054$&$0.063$&$0.076$\tabularnewline
$0.00$&$ 50$&&$0.020$&&$0.021$&$0.022$&$0.024$&$0.027$&&$0.021$&$0.022$&$0.024$&$0.028$\tabularnewline
$0.00$&$100$&&$0.010$&&$0.010$&$0.011$&$0.012$&$0.013$&&$0.010$&$0.011$&$0.012$&$0.014$\tabularnewline
$0.05$&$ 20$&&$0.223$&&$0.080$&$0.065$&$0.067$&$0.086$&&$0.080$&$0.064$&$0.077$&$0.066$\tabularnewline
$0.05$&$ 50$&&$0.144$&&$0.028$&$0.024$&$0.026$&$0.028$&&$0.028$&$0.024$&$0.030$&$0.026$\tabularnewline
$0.05$&$100$&&$0.118$&&$0.014$&$0.012$&$0.013$&$0.014$&&$0.014$&$0.012$&$0.015$&$0.013$\tabularnewline
$0.20$&$ 20$&&$1.761$&&$1.106$&$0.744$&$0.385$&$0.564$&&$1.104$&$0.727$&$0.304$&$0.280$\tabularnewline
$0.20$&$ 50$&&$1.571$&&$0.881$&$0.497$&$0.111$&$0.057$&&$0.879$&$0.477$&$0.082$&$0.042$\tabularnewline
$0.20$&$100$&&$1.509$&&$0.809$&$0.410$&$0.041$&$0.020$&&$0.807$&$0.385$&$0.026$&$0.019$\tabularnewline[10pt]
\multicolumn{14}{l}{\bfseries Moment Matching prior}\tabularnewline
$0.00$&$ 20$&&$0.050$&&$0.052$&$0.055$&$0.064$&$0.080$&&$0.052$&$0.055$&$0.063$&$0.074$\tabularnewline
$0.00$&$ 50$&&$0.020$&&$0.021$&$0.022$&$0.025$&$0.028$&&$0.021$&$0.022$&$0.025$&$0.028$\tabularnewline
$0.00$&$100$&&$0.010$&&$0.010$&$0.011$&$0.012$&$0.014$&&$0.010$&$0.011$&$0.012$&$0.014$\tabularnewline
$0.05$&$ 20$&&$0.223$&&$0.075$&$0.063$&$0.067$&$0.085$&&$0.075$&$0.063$&$0.067$&$0.076$\tabularnewline
$0.05$&$ 50$&&$0.144$&&$0.028$&$0.024$&$0.026$&$0.030$&&$0.028$&$0.024$&$0.026$&$0.029$\tabularnewline
$0.05$&$100$&&$0.118$&&$0.014$&$0.012$&$0.013$&$0.014$&&$0.014$&$0.012$&$0.013$&$0.015$\tabularnewline
$0.20$&$ 20$&&$1.761$&&$1.039$&$0.648$&$0.295$&$0.394$&&$1.043$&$0.655$&$0.286$&$0.290$\tabularnewline
$0.20$&$ 50$&&$1.571$&&$0.852$&$0.453$&$0.088$&$0.044$&&$0.853$&$0.443$&$0.078$&$0.043$\tabularnewline
$0.20$&$100$&&$1.509$&&$0.794$&$0.385$&$0.034$&$0.018$&&$0.794$&$0.365$&$0.025$&$0.018$\tabularnewline
\bottomrule
\end{tabular}
}
\end{center}
\end{table}

\begin{table}[!tbp]
\caption{Empirical MSEs of posterior means for $\sigma$\label{sim:mse_var}} 
\begin{center}
\resizebox{1\textwidth}{!}{
\begin{tabular}{rrcrcrrrrcrrrr}
\toprule
\multicolumn{2}{c}{\bfseries }&\multicolumn{1}{c}{\bfseries }&\multicolumn{1}{c}{\bfseries Bayes}&\multicolumn{1}{c}{\bfseries }&\multicolumn{4}{c}{\bfseries $R^{(\alpha)}$-posterior}&\multicolumn{1}{c}{\bfseries }&\multicolumn{4}{c}{\bfseries $\gamma$-posterior}\tabularnewline
\cline{4-4} \cline{6-9} \cline{11-14}
\multicolumn{1}{c}{$\varepsilon$}&\multicolumn{1}{c}{$n$}&\multicolumn{1}{c}{}&\multicolumn{1}{c}{$\alpha,\gamma \to 0.0$}&\multicolumn{1}{c}{}&\multicolumn{1}{c}{$\alpha=0.2$}&\multicolumn{1}{c}{$\alpha=0.3$}&\multicolumn{1}{c}{$\alpha=0.5$}&\multicolumn{1}{c}{$\alpha=0.7$}&\multicolumn{1}{c}{}&\multicolumn{1}{c}{$\gamma=0.2$}&\multicolumn{1}{c}{$\gamma=0.3$}&\multicolumn{1}{c}{$\gamma=0.5$}&\multicolumn{1}{c}{$\gamma=0.7$}\tabularnewline
\midrule
\multicolumn{14}{l}{\bfseries Uniform prior}\tabularnewline
$0.00$&$ 20$&&$0.033$&&$0.062$&$0.104$&$ 1.110$&$  8.455$&&$0.057$&$0.080$&$0.195$&$ 0.747$\tabularnewline
$0.00$&$ 50$&&$0.011$&&$0.015$&$0.019$&$ 0.034$&$  0.161$&&$0.015$&$0.017$&$0.025$&$ 0.036$\tabularnewline
$0.00$&$100$&&$0.005$&&$0.006$&$0.007$&$ 0.011$&$  0.018$&&$0.006$&$0.007$&$0.009$&$ 0.012$\tabularnewline
$0.05$&$ 20$&&$0.761$&&$0.424$&$0.471$&$ 7.528$&$ 37.358$&&$0.379$&$0.309$&$0.673$&$ 3.370$\tabularnewline
$0.05$&$ 50$&&$0.553$&&$0.095$&$0.051$&$ 0.066$&$  0.950$&&$0.087$&$0.040$&$0.035$&$ 0.047$\tabularnewline
$0.05$&$100$&&$0.482$&&$0.039$&$0.017$&$ 0.018$&$  0.031$&&$0.035$&$0.014$&$0.012$&$ 0.014$\tabularnewline
$0.20$&$ 20$&&$3.262$&&$4.185$&$5.830$&$55.081$&$138.264$&&$3.874$&$4.117$&$8.080$&$29.181$\tabularnewline
$0.20$&$ 50$&&$2.816$&&$2.706$&$2.229$&$ 1.895$&$ 27.513$&&$2.638$&$1.962$&$0.741$&$ 0.454$\tabularnewline
$0.20$&$100$&&$2.682$&&$2.405$&$1.704$&$ 0.483$&$  0.506$&&$2.372$&$1.526$&$0.146$&$ 0.038$\tabularnewline[10pt]
\multicolumn{14}{l}{\bfseries Reference prior}\tabularnewline
$0.00$&$ 20$&&$0.027$&&$0.030$&$0.033$&$ 0.040$&$  0.059$&&$0.030$&$0.032$&$0.041$&$ 0.058$\tabularnewline
$0.00$&$ 50$&&$0.010$&&$0.011$&$0.012$&$ 0.015$&$  0.017$&&$0.011$&$0.012$&$0.015$&$ 0.020$\tabularnewline
$0.00$&$100$&&$0.005$&&$0.006$&$0.006$&$ 0.007$&$  0.008$&&$0.006$&$0.006$&$0.007$&$ 0.009$\tabularnewline
$0.05$&$ 20$&&$0.611$&&$0.153$&$0.083$&$ 0.068$&$  0.101$&&$0.145$&$0.073$&$0.050$&$ 0.054$\tabularnewline
$0.05$&$ 50$&&$0.504$&&$0.054$&$0.023$&$ 0.019$&$  0.021$&&$0.051$&$0.021$&$0.017$&$ 0.021$\tabularnewline
$0.05$&$100$&&$0.459$&&$0.027$&$0.011$&$ 0.009$&$  0.010$&&$0.025$&$0.010$&$0.008$&$ 0.010$\tabularnewline
$0.20$&$ 20$&&$2.731$&&$2.283$&$1.624$&$ 0.941$&$  0.982$&&$2.232$&$1.482$&$0.548$&$ 0.304$\tabularnewline
$0.20$&$ 50$&&$2.633$&&$2.165$&$1.330$&$ 0.341$&$  0.215$&&$2.140$&$1.218$&$0.171$&$ 0.048$\tabularnewline
$0.20$&$100$&&$2.595$&&$2.158$&$1.268$&$ 0.144$&$  0.070$&&$2.143$&$1.144$&$0.046$&$ 0.014$\tabularnewline[10pt]
\multicolumn{14}{l}{\bfseries Moment Matching prior}\tabularnewline
$0.00$&$ 20$&&$0.026$&&$0.028$&$0.031$&$ 0.040$&$  0.063$&&$0.028$&$0.032$&$0.042$&$ 0.054$\tabularnewline
$0.00$&$ 50$&&$0.010$&&$0.011$&$0.012$&$ 0.015$&$  0.019$&&$0.011$&$0.012$&$0.015$&$ 0.020$\tabularnewline
$0.00$&$100$&&$0.005$&&$0.006$&$0.006$&$ 0.007$&$  0.009$&&$0.006$&$0.006$&$0.007$&$ 0.009$\tabularnewline
$0.05$&$ 20$&&$0.525$&&$0.105$&$0.058$&$ 0.046$&$  0.052$&&$0.104$&$0.057$&$0.048$&$ 0.056$\tabularnewline
$0.05$&$ 50$&&$0.470$&&$0.043$&$0.020$&$ 0.017$&$  0.018$&&$0.041$&$0.019$&$0.017$&$ 0.021$\tabularnewline
$0.05$&$100$&&$0.443$&&$0.023$&$0.010$&$ 0.008$&$  0.009$&&$0.022$&$0.009$&$0.008$&$ 0.010$\tabularnewline
$0.20$&$ 20$&&$2.411$&&$1.809$&$1.132$&$ 0.441$&$  0.186$&&$1.816$&$1.137$&$0.461$&$ 0.385$\tabularnewline
$0.20$&$ 50$&&$2.507$&&$1.974$&$1.120$&$ 0.222$&$  0.082$&&$1.971$&$1.065$&$0.153$&$ 0.054$\tabularnewline
$0.20$&$100$&&$2.532$&&$2.065$&$1.148$&$ 0.106$&$  0.040$&&$2.059$&$1.054$&$0.043$&$ 0.015$\tabularnewline
\bottomrule
\end{tabular}
}
\end{center}
\end{table}




Figure \ref{contami:fig} shows the results of empirical bias and MSE of posterior means of $\mu$ and $\sigma$ under the uniform, reference and moment matching priors when $\nu = 6$ (fixed) and the contamination ratio $\ep$ varies from $0.00$ to $0.30$. In all cases,we can find that the standard posterior means (i.e. cases $\alpha,\gamma=0$) do not work well. For estimation of mean parameter $\mu$, $R^{(\alpha)}$- and $\gamma$-posterior means seems to be reasonable for the value of $\ep$ between $0.0$ and $0.20$. In particular, $\gamma$-posterior means under reference and moment matching priors has better performance even if $\ep=0.30$. For estimation of variance parameter $\sigma$, $R^{(\alpha)}$-posterior means under uniform prior is larger bias and MSE than other methods. The $\gamma$-posterior mean with $\gamma=1.0$ still may be better than other competitors for any $\ep \in [0, 0.30]$. For $\alpha,\gamma=0.5$, $R^{(\alpha)}$- and $\gamma$-posterior means seems to be comparable. 

Figure \ref{mu:fig} also presents the results of empirical biase and MSE  of posterior means of $\mu$ and $\sigma$ under the same priors as Figure \ref{contami:fig} when the contamination ratio is $\ep = 0.20$ (fixed) and $\nu$ varies from $0.0$ to $10.0$. For estimation of mean parameter $\mu$ in Figure \ref{mu:fig}, empirical bias and MSE for robust estimators seems to be nice regardless of $\nu$ except for the case of $R^{(\alpha)}$-posterior under uniform prior. Although we can find that some differences appear near $\nu=4$, $\gamma$-posterior means with $\gamma=1.0$ have better performance for estimation of both mean $\mu$ and variance $\sigma$ for all $\nu \in [0,10]$. 

In these simulation studies, the $\gamma$-posterior mean under reference and moment matching priors seems to have better performance for joint estimation of $(\mu,\sigma)$ in most scenarios. Although we provide the results for the univariate normal distribution, the other distribution (including the multivariate distribution) should be also considered in the future.




\begin{figure}[!tb]
\centering
\includegraphics[width=15cm,clip]{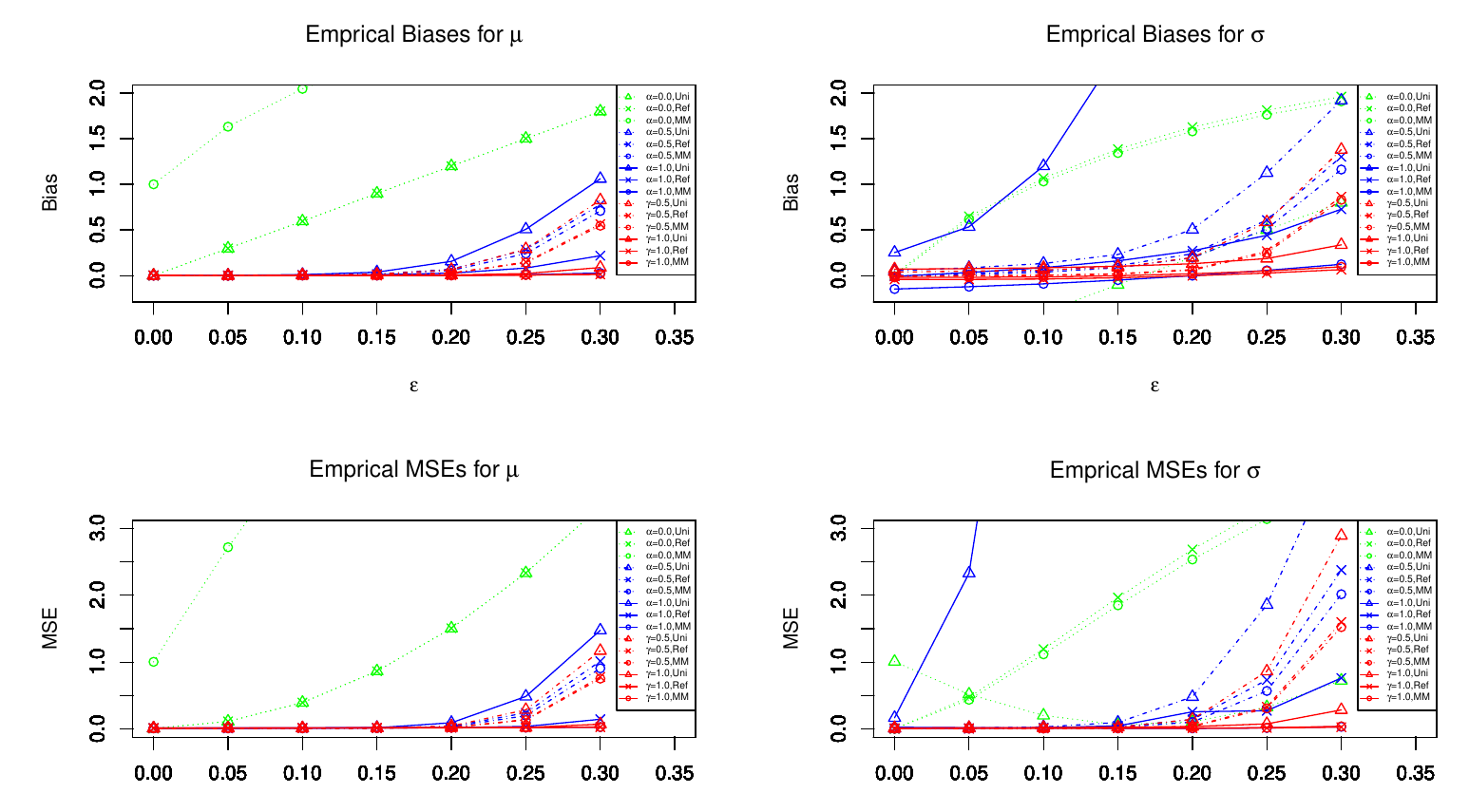}
\caption{The horizontal axis is the contamination ratio $\ep$. The red lines show empirical bias and MSE of $\gamma$-posterior means under the three priors when $n = 100$ and $\nu = 6$. Similarly, the blue and green lines show that of $R^{(\alpha)}$-posterior and ordinary posterior means, respectively. Uniform, reference, and moment matching priors are denoted by ``Uni", ``Ref", and ``MM", respectively.}
\label{contami:fig}
\end{figure}

\begin{figure}[!tb]
\centering
\includegraphics[width=15cm,clip]{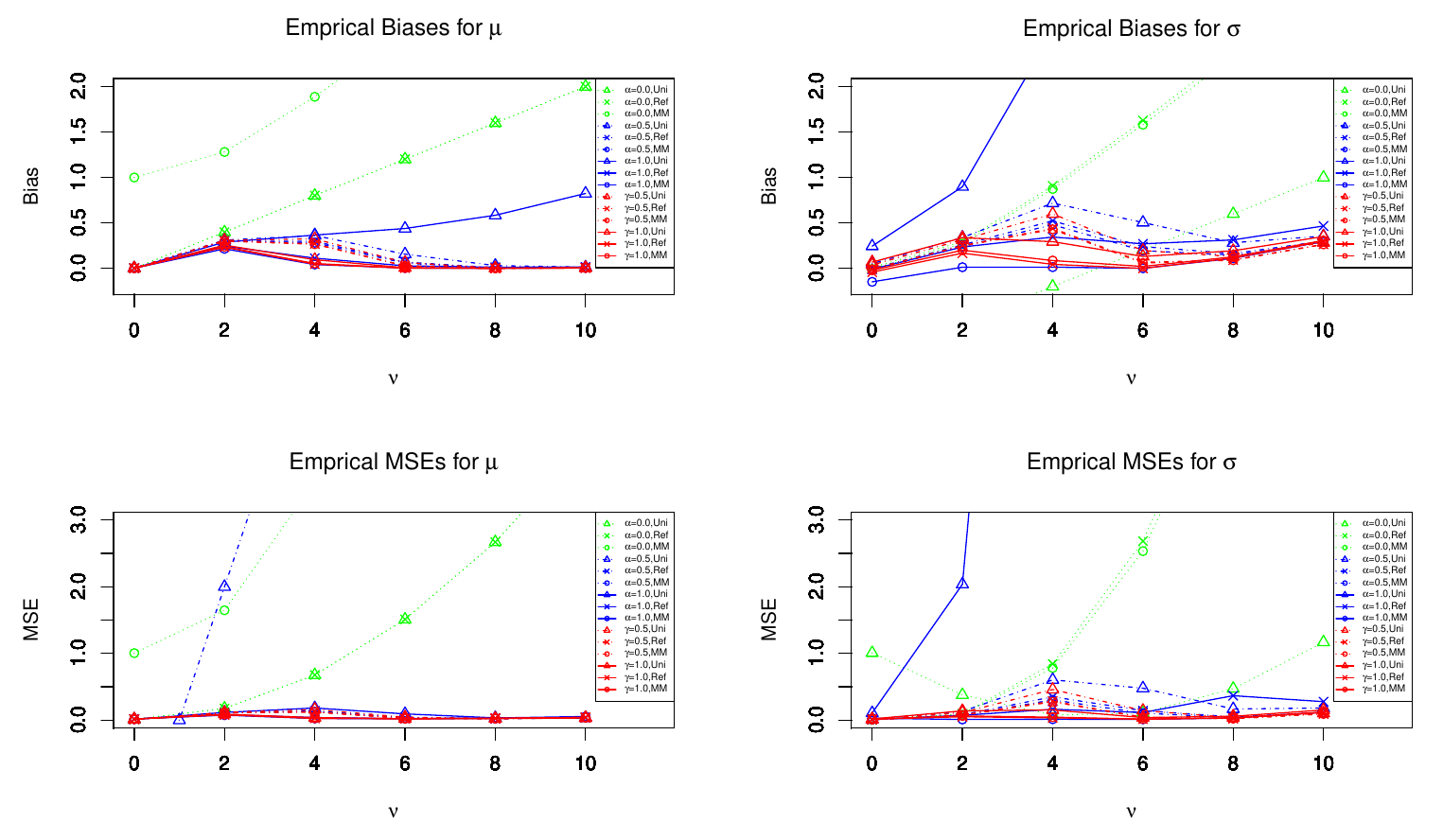}
\caption{The horizontal axis is the location parameter $\nu$ of contamination distribution. The red lines show empirical bias and MSE of $\gamma$-posterior means under the three priors when $n = 100$ and $\ep = 0.20$. Similarly, the blue and green lines show that of $R^{(\alpha)}$-posterior and ordinary posterior means, respectively. Uniform, reference, and moment matching priors are denoted by ``Uni", ``Ref", and ``MM", respectively. }
\label{mu:fig}
\end{figure}

\subsection{Selection of tuning parameters}
The selection of a tuning parameter $\gamma$ (or $\alpha$) is very challenging and, to the best of our knowledge, there is no optimal choice of $\gamma$. The tuning parameter $\gamma$ controls the degree of robustness, that is, if we set large $\gamma$, we obtain the higher robustness. However, there is a trade-off between robustness and efficiency of estimators. One of solutions for this problem is to use the asymptotic relative efficiency (ARE) (see e.g. \cite{ghosh2016robust}). It should be noted that \cite{ghosh2016robust} only deals with a one-parameter case. 
In general, the asymptotic relative efficiency of the robust posterior mean $\hat{\bm{\theta}}^{(\gamma)}$ of $p$-dimensional parameter $\bm{\theta}$ relative to the usual posterior mean $\hat{\bm{\theta}}$ is defined by
\begin{align*}
\mathrm{ARE}(\hat{\bm{\theta}}^{(\gamma)},\hat{\bm{\theta}}):=\left(\frac{\det\left(V(\bm{\theta})\right)}{\det\left(V^{(\gamma)}(\bm{\theta})\right)}\right)^{1/p}
\end{align*}
(see e.g. \cite{serfling1980book}). This is the ratio of determinants of the covariance matrices, raised to the power of $1/p$, where $p$ is the dimension of the parameter $\bm{\theta}$. We now calculate the $\mathrm{ARE}(\hat{\bm{\theta}}^{(\gamma)},\hat{\bm{\theta}})$ in our simulation setting. After some calculations, the asymptotic relative efficiency is given by
\begin{align*}
\mathrm{ARE}(\hat{\bm{\theta}}^{(\gamma)},\hat{\bm{\theta}})=\left(\frac{2}{(1+\gamma)^6(1+2\gamma)(2+4\gamma+3\gamma^2)} \right)^{1/2} =: h(\gamma)
\end{align*}
for $\gamma>0$. We note that it holds $h(\gamma)\to 1 $ as $\gamma\to 0$. Hence, we may be able to choose $\gamma$ to allow for the small inflation of the efficiency. For example, if we require the value of the asymptotic relative efficiency $\mathrm{ARE}=0.95$, we may choose the value of $\gamma$ as the solution of the equation $h(\gamma)=0.95$ (see Table \ref{ARE:tab}). The curve of the function $h(\gamma)$ is also given in Figure \ref{ARE:fig}. Several authors provide methods for the selection of the tuning parameters (e.g. \cite{Warwick2005}, \cite{sugasawa2020robust} and \cite{Basak2020}). \cite{fujisawa2008robust} focused on the reduction of the latent bias of estimator and they recommend to set $\gamma = 1$ for normal mean-variance estimation problem, but it seems to be unreasonable in terms of the asymptotic relative efficiency (see Table \ref{ARE:tab} and Figure \ref{ARE:fig}). To the best of our  knowledge, there is no methods which are robust and efficient under heavy contamination setting. Hence, other methods which have higher efficiency under heavy contamination should be considered in the future.

\begin{table}[!htb]

\caption{The value of $\gamma$ and the corresponding asymptotic relative efficiency}
\begin{center}
\begin{tabular}{|c|c|c|c|c|} \hline
$\gamma$ & $0.01$ & $0.1$ & $0.3$ & $0.5$ \\ \hline
ARE & 0.951489 & 0.6222189 & 0.2731871 & 0.1359501\\ \hline
\end{tabular}
\label{ARE:tab}
\end{center}
\end{table}

\begin{figure}[!htb]
\centering
\includegraphics[width=7cm,clip]{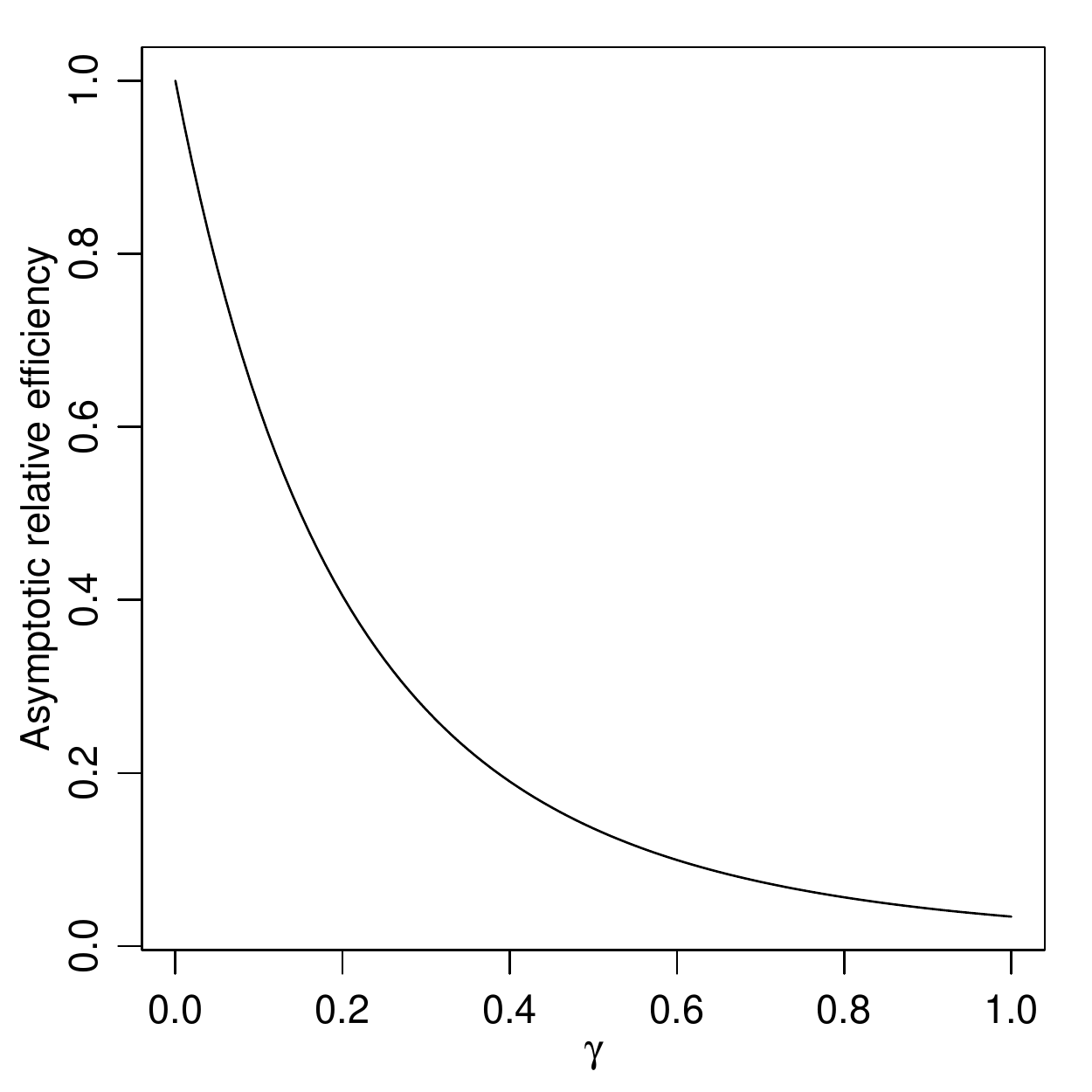}
\caption{The curve of the asymptotic relative efficiency for normal mean and variance estimation under the $\gamma$-posterior.}
\label{ARE:fig}
\end{figure}

\section{Concluding remarks}

We considered objective priors for divergence-based robust Bayesian estimation. In particular, we proved that reference and moment matching priors under quasi-posterior based on the $\gamma$-divergence are robust against unknown quantities in a data generating distribution. The performance of the corresponding posterior means is illustrated through some simulation studies. However, the proposed objective priors are often improper, and to show the posterior propriety for them remains as the future research. Our results should be extended to other settings. For example, Kanamori and Fujisawa \cite{kanamori2015robust} proposed the estimation of the contamination ratio using an unnormalized model. To examine such problem from the Bayesian perspective is also challenging because there is a problem of how to set a prior distribution for an unknown contamination ratio. Furthermore, it would be also interesting to consider an optimal data-dependent choice of tuning parameter $\gamma$.

\section*{Acknowledgments}
This work was supported by JSPS Grant-in-Aid for Early-Career Scientists Grant Number JP19K14597 and JSPS Grant-in-Aid for Young Scientists (B) Grant Number JP17K14233. 
\bibliographystyle{apalike} 
\bibliography{References.bib}

\def\thesection{Appendix}
\def\thesubsection{A.\arabic{subsection}}
\section{}
\setcounter{equation}{0}
\def\theequation{A.\arabic{equation}}
\def\thelem{A.\arabic{lem}}
\allowdisplaybreaks[1]
\subsection{Some derivative functions}
We now put $\ell(x) = \log f_{\bm{\theta}}(x)$, $\ell_i(x) = \partial_i\log f_{\bm{\theta}}(x)$, $\ell_{ij}(x) = \partial_i\partial_j\log f_{\bm{\theta}}(x)$ and $\ell_{ijk}(x) = \partial_i\partial_j\partial_k\log f_{\bm{\theta}}(x)$, and let a norm $\|\cdot\|_p : L_p(\Omega) \to \mathbb{R}$ be defined by
\[
 \|h\|_{p} = \left(\int_{\Omega} |h(x)|^{p} \dd x\right)^{1/p}. 
\] 
We then obtain derivative functions of $q^{(\alpha)}(x_j ; \bm{\theta})$ with respect to $\bm{\theta}$ as follows:
\begin{align*}
\begin{split}
\partial_i q^{(\alpha)}(x; \bm{\theta})  =&  f_{\bm{\theta}}(x)^{\alpha} \ell_i(x)  - \int_{\Omega} f_{\bm{\theta}}(y)^{\alpha + 1} \ell_i(y)\dd y, \\
\partial_i \partial_j q^{(\alpha)}(x; \bm{\theta}) =& f_{\bm{\theta}}(x)^{\alpha} \left\{ \alpha\ell_i(x)\ell_j(x) + \ell_{ij}(x)\right\} \\
& - \int_{\Omega} f_{\bm{\theta}}(y)^{\alpha + 1}\left\{ (\alpha + 1)\ell_i(y)\ell_i(y)+ \ell_{ij}(y) \right\} \dd y, \\
\partial_i \partial_j\partial_k q^{(\alpha)}(x; \bm{\theta})  =& f_{\bm{\theta}}(x)^{\alpha} \left\{ \alpha^2\ell_i(x)\ell_j(x)\ell_k(x)+ \ell_{ijk}(x)\right.\\
&- \alpha\left(\ell_k(x)\ell_{ij}(x) + \ell_i(x)\ell_{jk}(x) \left. + \ell_j(x)\ell_{ik}(x)\right)  \right\}\\
& - \int_{\Omega} f_{\bm{\theta}}(y)^{\alpha + 1}\left\{ (\alpha + 1)^2\ell_i(y)\ell_j(y)\ell_k(y)+ \ell_{ijk}(y)  \right.\\
& \left. + (\alpha+ 1) \left\{\ell_k(y)\ell_{ij}(y)+ \ell_i(y)\ell_{jk}(y) + \ell_j(y)\ell_{ik}(y)\right\}\right\} \dd y.  \\
\end{split}
\end{align*}
Similarly, we obtain derivative functions of  $q^{(\gamma)}(x_j ; \bm{\theta})$ as follows: 
\begin{align*}
\begin{split}
\partial_i q^{(\gamma)}(x; \bm{\theta})  =& \frac{f_{\bm{\theta}}(x)^{\gamma}}{\|f_{\bm{\theta}}\|_{\gamma+1}^{\gamma}}\ell_i(x) - \frac{f_{\bm{\theta}}(x)^{\gamma}}{\|f_{\bm{\theta}}\|_{\gamma+1}^{1+2\gamma}}\int_{\Omega} f_{\bm{\theta}}(y)^{\gamma+1}\ell_i(y) \dd y, \\
\partial_i\partial_j q^{(\gamma)}(x_j; \bm{\theta})  =&  \frac{f_{\bm{\theta}}(x)^{\gamma}}{\|f_{\bm{\theta}}\|_{\gamma+1}^{\gamma}}\left(\gamma\ell_i(x)\ell_j(x) + \ell_{ij}(x)\right)\\
& - \gamma  \frac{f_{\bm{\theta}}(x)^{\gamma}}{\|f_{\bm{\theta}}\|_{\gamma+1}^{1+2\gamma}}\left(\ell_j(x)\int_{\Omega} f_{\bm{\theta}}(y)^{\gamma+1}\ell_i(y)\dd y +\ell_i(x)\int_{\Omega} f_{\bm{\theta}}(y)^{\gamma+1}\ell_j(y)\dd y \right)\\
& + (1+2\gamma) \frac{f_{\bm{\theta}}(x)^{\gamma}}{\|f_{\bm{\theta}}\|_{\gamma+1}^{2+3\gamma}}\int_{\Omega} f_{\bm{\theta}}(y)^{\gamma+1}\ell_i(y)\dd y\int_{\Omega} f_{\bm{\theta}}(y)^{\gamma+1}\ell_j(y)\dd y\\
& - \frac{f_{\bm{\theta}}(x)^{\gamma}}{\|f_{\bm{\theta}}\|_{\gamma+1}^{1+2\gamma}}\left(\int_{\Omega} f_{\bm{\theta}}(y)^{\gamma+1}s_{ij}(y)\dd y\right), \\
\partial_i\partial_j\partial_k q^{(\gamma)}(x; \bm{\theta})  =&  \frac{f_{\bm{\theta}}(x)^{\gamma}}{\|f_{\bm{\theta}}\|_{\gamma+1}^{\gamma}}\left(\gamma^2\ell_i(x)\ell_j(x)\ell_k(x) + \ell_{ijk}(x)\right)\\
& +\gamma \frac{f_{\bm{\theta}}(x)^{\gamma}}{\|f_{\bm{\theta}}\|_{\gamma+1}^{\gamma}}\left(\ell_{k}(x)\ell_{ij}(x) + \ell_j(x)\ell_{ik}(x) + \ell_i(x)\ell_{jk}(x)\right)\\
&  -  \frac{f_{\bm{\theta}}(x)^{\gamma}}{\|f_{\bm{\theta}}\|_{\gamma+1}^{1+2\gamma}}\left(\gamma^2\ell_j(x)\ell_k(x) + \gamma \ell_{jk}(x)\right)\int_{\Omega} f_{\bm{\theta}}(y)^{\gamma+1}\ell_i(y)\dd y\\
&   -  \frac{f_{\bm{\theta}}(x)^{\gamma}}{\|f_{\bm{\theta}}\|_{\gamma+1}^{1+2\gamma}}\left(\gamma^2\ell_i(x)\ell_k(x) + \gamma \ell_{ik}(x)\right)\int_{\Omega} f_{\bm{\theta}}(y)^{\gamma+1}\ell_j(y)\dd y\\
&   -  \frac{f_{\bm{\theta}}(x)^{\gamma}}{\|f_{\bm{\theta}}\|_{\gamma+1}^{1+2\gamma}}\left(\gamma^2\ell_i(x)\ell_j(x) + \gamma \ell_{ij}(x)\right)\int_{\Omega} f_{\bm{\theta}}(y)^{\gamma+1}\ell_k(y)\dd y\\
& + (1+\gamma)(1+2\gamma) \frac{f_{\bm{\theta}}(x)^{\gamma}}{\|f_{\bm{\theta}}\|_{\gamma+1}^{2+3\gamma}}\ell_k(x)\int_{\Omega} f_{\bm{\theta}}(y)^{\gamma+1}\ell_i(y)\dd y\int_{\Omega} f_{\bm{\theta}}(y)^{\gamma+1}\ell_j(y)\dd y\\
& + (1+\gamma)(1+2\gamma) \frac{f_{\bm{\theta}}(x)^{\gamma}}{\|f_{\bm{\theta}}\|_{\gamma+1}^{2+3\gamma}}\ell_j(x)\int_{\Omega} f_{\bm{\theta}}(y)^{\gamma+1}\ell_i(y)\dd y\int_{\Omega} f_{\bm{\theta}}(y)^{\gamma+1}\ell_k(y)\dd y\\
& + (1+\gamma)(1+2\gamma) \frac{f_{\bm{\theta}}(x)^{\gamma}}{\|f_{\bm{\theta}}\|_{\gamma+1}^{2+3\gamma}}\ell_i(x)\int_{\Omega} f_{\bm{\theta}}(y)^{\gamma+1}\ell_j(y)\dd y\int_{\Omega} f_{\bm{\theta}}(y)^{\gamma+1}\ell_k(y)\dd y\\
&- \gamma  \frac{f_{\bm{\theta}}(x)^{\gamma}}{\|f_{\bm{\theta}}\|_{\gamma+1}^{1+2\gamma}}\ell_{k}(x)\int_{\Omega} f_{\bm{\theta}}(y)^{\gamma+1}s_{ij}(y)\dd y\\
 & - \gamma  \frac{f_{\bm{\theta}}(x)^{\gamma}}{\|f_{\bm{\theta}}\|_{\gamma+1}^{1+2\gamma}}\ell_{j}(x)\int_{\Omega} f_{\bm{\theta}}(y)^{\gamma+1}s_{ik}(y)\dd y\\
 & - \gamma  \frac{f_{\bm{\theta}}(x)^{\gamma}}{\|f_{\bm{\theta}}\|_{\gamma+1}^{1+2\gamma}}\ell_{i}(x)\int_{\Omega} f_{\bm{\theta}}(y)^{\gamma+1}s_{jk}(y)\dd y\\
&\ -  \frac{(1+2\gamma)(2+3\gamma)}{\|f_{\bm{\theta}}\|_{\gamma+1}^{3+4\gamma}}f_{\bm{\theta}}(x)^{\gamma}S_{ijk} \\
& - \frac{f_{\bm{\theta}}(x)^{\gamma}}{\|f_{\bm{\theta}}\|_{\gamma+1}^{1+2\gamma}}\int_{\Omega} f_{\bm{\theta}}(y)^{\gamma+1}\left\{(\gamma + 1)^2\ell_i(y)\ell_j(y)\ell_k(y)+\ell_{ijk}(y) \right\}\dd y\\
& - \frac{f_{\bm{\theta}}(x)^{\gamma}}{\|f_{\bm{\theta}}\|_{\gamma+1}^{1+2\gamma}}\int_{\Omega} f_{\bm{\theta}}(y)^{\gamma+1}s_{ijk}(y)\dd y,   
\end{split}
\end{align*}
where 
\begin{align*}
s_{ij}(y) = &(\gamma+1)\ell_{i}(y)\ell_{j}(y) + \ell_{ij}(y), \\
s_{ijk}(y) =& (\gamma + 1)\{\ell_{ij}(y)\ell_k(y)+\ell_{j}(y)\ell_{ik}(y) + \ell_{i}(y)\ell_{jk}(y)\}, \\
S_{ijk} = &\int_{\Omega} f_{\bm{\theta}}(y)^{\gamma+1}\ell_i(y)\dd y\int_{\Omega} f_{\bm{\theta}}(y)^{\gamma+1}\ell_j(y)\dd y\int_{\Omega} f_{\bm{\theta}}(y)^{\gamma+1}\ell_k(y)\dd y. 
\end{align*}

\end{document}